\documentclass[
 amsmath,amssymb
]{revtex4}

\usepackage{graphicx}
\usepackage{dcolumn}
\usepackage{bm}
\usepackage{natbib}
\usepackage{color}
\usepackage{subcaption}
\numberwithin{equation}{section}

\newtheorem{theorem}{Theorem}[section]
\newtheorem{proposition}{Proposition}[section]
\newtheorem{corollary}{Corollary}[section] 

\newtheorem{lemma}[proposition]{Lemma}
\newenvironment{proof}{{\noindent\it Proof:}\quad}{\hfill $\square$\par}  
\newtheorem{rem}{Remark}

\newcommand{\D}{\mathrm{D}}

\newcommand\nvec{\mathbf{n}}

\newcommand\xhat{\hat{\mathbf{x}}}
\newcommand\yhat{\hat{\mathbf{y}}}
\newcommand\zhat{\hat{\mathbf{z}}}
\newcommand\Qvec{\mathbf{Q}}
\newcommand\Pvec{\mathbf{P}}

\begin{document}
\title{Multistability for a Reduced Nematic Liquid Crystal Model in the Exterior of 2D Polygons}
\author{Yucen Han and Apala Majumdar}
\affiliation{Department of Mathematics and Statistics, University of Strathclyde, G1 1XQ, United Kindom.}

\begin{abstract}
We study nematic equilibria in an unbounded domain, with a two-dimensional regular polygonal hole with $K$ edges, in a reduced Landau--de Gennes framework. 
This complements our previous work on the ``interior problem" for nematic equilibria confined inside regular polygons (SIAM Journal on Applied Mathematics, 80(4):1678–1703, 2020). The two essential dimensionless model parameters are $\lambda$-the ratio of the edge length of polygon hole to the nematic correlation length, and an additional degree of freedom,  $\gamma^*$-the nematic director at infinity. In the $\lambda\to 0$ limit, the limiting profile has two interior point defects outside a generic polygon hole, except for a triangle and a square. For a square hole, the limiting profile has either no interior defects or two line defects depending on $\gamma^*$, and for a triangular hole, there is a unique interior point defect outside the hole. In the $\lambda\to\infty$ limit, there are at least $\binom{K}{2}$ stable states 
and the multistability is enhanced by $\gamma^*$, compared to the interior problem. Our work offers new insights into how to tune the existence, location, and dimensionality of defects. 
\end{abstract}

\maketitle

\section{Introduction}
\label{sec:intro}
Nematic liquid crystals (NLCs) are classical examples of partially ordered materials or viscoelastic anisotropic materials, with long-range orientational ordering \cite{de1993physics}. The NLC molecules are typically asymmetric in shape e.g. rod-like, and they tend to align along locally preferred directions, referred to as \emph{nematic directors} in the literature \cite{de1993physics}. The optical, mechanical and rheological NLC responses are direction-dependent, with distinguished responses along the nematic directors. Indeed, the anisotropic NLC responses to external stimuli, interfaces and boundaries make them soft, self-organising and the cornerstone of several NLC applications in science and engineering \cite{phillips2011texture}. Nematics in confinement have attracted substantial scientific interest in the academic and industrial sectors \cite{lagerwallreview}. In fact, nematics inside a planar cell are the building block of the celebrated twisted nematic display, and contemporary work has focused on the tremendous potential of NLCs for multistable systems i.e. NLC systems with multiple observable states without applied fields, such that each observable state offers a distinct mode of functionality \cite{zbd}. Nematic defects play a key role in multistability, where a nematic defect is a point/line/surface where the nematic director cannot be uniquely defined \cite{de1993physics}. Nematic defects often organise the space of stable states in multistable systems in the sense that we can classify the stable states in terms of the nature, multiplicity and locations of defects. Additionally, defects have distinct optical signatures and can act as binding sites or ``special'' sites in materials design or applications. The delicate interplay between geometric frustration, boundary effects, material properties and defects in multistability leads to a suite of challenging mathematical questions at the interface of applied topology, nonlinear partial differential equations and scientific computation (to name a few). Equally, it gives new inroads into engineered soft materials, topological materials or meta-materials which could find new applications in photonics, robotics and artificial intelligence \cite{lagerwallcholestericshells}.

In this paper, we focus on the mathematical modelling and numerical computation of NLC equilibria outside regular two-dimensional (2D) polygons with homeotropic boundary conditions i.e. the nematic director is constrained to be normal to the polygon boundary. We work in a reduced 2D Landau--de Gennes framework, details of which are given in the next section, and this can capture the nematic directors and the nematic defects in a 2D setting, along with informative insights into how multistability can be tailored by the shape and size of the polygon. This toy mathematical problem models a single colloidal particle, in the shape of a 2D polygon, suspended in an extended NLC medium, which is of both experimental and theoretical interest \cite{Mu2008Self}, \cite{smalyukh1, smalyukh2}, \cite{gupta2005texture}, \cite{phillips2011texture}. In \cite{Mu2008Self, smalyukh1, smalyukh2}, the authors fabricate almost 2D platelets of different polygonal shapes suspended in a NLC medium. Using advanced microfabrication techniques and optical methods, they can track the director profiles and the associated defects. The authors observe multiple types of defects: dipoles, Saturn rings, and various linked or entangled defect loop lines intertwining arrays of embedded inclusions. In fact, in \cite{Mu2008Self}, the authors use optical tweezers to manipulate the defect lines, to link them or disentangle them, and in doing so, create various exotic knotted defect patterns. In all cases, the observed states and their defect patterns strongly depend on the geometry and orientation of colloidal particle(s) and their boundary effects, offering excellent examples of organic self-assembled structures in NLC media.

In \cite{han2020pol}, we study the ``interior'' problem of multistability for NLCs confined to a regular 2D polygon with tangent boundary conditions i.e. the nematic director is tangent to the polygon edges, in a reduced Landau--de Gennes (LdG) framework. We use conformal mapping techniques and methods from elliptic partial differential equations to study nematic equilibria in two distinguished limits, phrased in terms of a dimensionless variable $\lambda$. The variable, $\lambda$, is the ratio of two length scales - the ratio of the polygon edge length to the nematic correlation length, which is a material-dependent length scale related to typical defect sizes. In the $\lambda \to 0$ limit, the limiting profile is unique with a single isolated point defect at the centre of the polygon, coined as the \emph{Ring} solution. The only exceptions are the triangle and the square, which are dealt with separately. In fact, the limiting profile for a triangle has an isolated fractional point defect at the centre whereas the unique limiting profile for a square domain is labelled as the Well Order Reconstruction Solution (WORS), first reported in \cite{kraljmajumdar2014}. The WORS is a distinctive profile  with two orthogonal defect lines along the square diagonals (also see \cite{kraljmajumdar2014}). In the $\lambda \to \infty$ limit, we use combinatorial arguments to demonstrate multistability i.e. there are at least $\binom{K}{2}$ distinct NLC equilibria on a $K$-polygon with $K$ edges. In related papers, we comprehensively study NLC solution landscapes on 2D polygons with tangent boundary conditions \cite{hannonlinearity2020, hanproceedings2021}.

We study the complementary ``exterior'' problem in this paper: asymptotic and numerical studies of NLC equilibria outside a regular polygonal hole, immersed in $\mathbb{R}^2$. This problem, and related problems, have received some analytical interest although systematic studies are missing. For example, in \cite{phillips2011texture}, the authors study stable NLC profiles outside a square hole with homeotropic boundary conditions. They numerically observe string defects (line defects) pinned to square edges, defects at square vertices and interior point defects, depending on the temperature and square size. In \cite{wang2017topological,wang2018formation}, the authors numerically investigate the possible structures of NLCs with one, two and multiple spherical inclusions. In \cite{gupta2005texture}, the authors model NLCs with multiple spherical inclusions and numerically investigate how the defect set depends on the spatial organisation and properties of the spherical particles, along with those of the ambient NLC media. In \cite{Bronsard2016minimizers}, the authors rigorously analyse NLC equilibria outside a spherical particle with homeotropic boundary conditions, in a three-dimensional (3D) LdG framework. They obtain elegant limiting profiles in the small particle and large particle limit, and in fact, produce an analytic expression for the celebrated Saturn ring solution with a distinct defect loop around the spherical particle. 

In this paper, we focus on the effects of shape and size of the polygonal hole on the corresponding NLC equilibria, in a reduced LdG model. The methodology follows that of \cite{han2020pol}, the key difference being the extra degree of freedom rendered by the far-field boundary conditions, away from the polygon boundary. As with the interior problem, we compute limiting profiles for the stable NLC equilibria in the $\lambda \to 0$ and $\lambda \to \infty$ limits, where $\lambda$ has the same interpretation as in \cite{han2020pol}, accompanied by supplementary numerical results. The stable equilibria are modelled by local or global energy minimisers of the reduced LdG free energy, which in turn are solutions of the associated Euler-Lagrange equations that are a system of two coupled nonlinear partial differential equations.  In the $\lambda \to 0$ limit, there is a unique NLC equilibrium or equivalently, a unique minimiser of the reduced LdG free energy. However, the limiting profiles are more varied compared to the interior problem. For a square hole, we can observe either line defects or point defects at the square vertices, depending on the far-field condition, in the $\lambda \to 0$ limit. There is qualitative agreement with the numerical results in \cite{phillips2011texture}. In general, the locations of the defects for the unique limiting profile depend on the far-field condition. Using a result from Ginzburg--Landau theory in \cite{baumanowensphillips}, we show that there are exactly two interior point defects for a generic polygonal hole, $E_K$ with $K$ edges and $K>4$, and the location of these defects can be tuned with the far-field condition.
In the $\lambda \to \infty$ limit, we provide a simple estimate for the number of stable NLC equilibria using combinatorial arguments, and multistability is enhanced compared to the interior problem. This is because the exterior problem has lesser symmetry than the interior problem, due to the far-field conditions. For example, on the interior of a square domain, there are two rotationally equivalent diagonal solutions for which the NLC profile is approximately aligned along the square diagonal. We lose this equivalence for the exterior problem since the far-field condition breaks the symmetry. 

Applied mathematics focuses on the development of new mathematical methods, and equally elegant applications of known methods to new settings. Our work falls into the second category, where we largely build on previous work, and use techniques from complex analysis, Ginzburg--Landau theory for superconductivity, symmetry results and combinatorial arguments to analyse limiting profiles, complemented by numerical studies to support the theory. In doing so, we demonstrate how geometric frustration and nematic defects go hand in hand for tailored multistability, and this is a good forward step for rigorous mathematical studies of NLC solution landscapes in complex geometries with voids, mixed boundary conditions and in some cases, multiple order parameters \cite{hanpre2021}.

The paper is organised as follows. In Section~\ref{sec:theory}, we review the reduced 2D LdG framework for modelling NLCs in 2D scenarios. In Section~\ref{sec:lambda0}, we focus on the $\lambda \to 0$ limit of minimisers of the reduced Landau--de Gennes free energy. We use the Schwarz--Christoffel mapping to define an associated boundary-value problem on the unit disc, for each regular polygonal hole, and this boundary-value problem is explicitly solved. The defect set is tracked analytically, along with its dependence on $K$ and the far-field condition. In Section~\ref{sec:infinity}, we shift focus to the $\lambda \to \infty$ limit and the limiting problem is captured by the Laplace problem for an angle in the plane, with Dirichlet boundary conditions. This angle models the 2D nematic director. The Dirichlet boundary conditions for the angle are not uniquely defined, and this leads to multistability in this limit. We present illustrative numerical examples for a square and a hexagon, and conclude in Section~\ref{sec:conclusion} with a summary and some perspectives.
\section{Theoretical framework}
\label{sec:theory}
The Landau--de Gennes (LdG) theory is a celebrated continuum theory for nematic liquid crystals \cite{de1993physics}, and was indeed one of the reasons for awarding the Nobel Prize for physics to Pierre-Gilles de Gennes in 1991. The LdG theory is a phenomenological theory that assumes macroscopic quantities of interest vary slowly
on the molecular length scale, and is based on the crucial concept of a LdG order parameter, which is a macroscopic measure of the NLC anisotropy or orientational ordering. The LdG order parameter, known as the $\mathbf{Q}$-tensor order
parameter is a symmetric, traceless $3\times 3$ matrix with five degrees of freedom. The nematic director is often interpreted as the eigenvector of the LdG $\mathbf{Q}$-tensor with the largest positive eigenvalue \cite{han2020pol}. Consider a three-dimensional (3D) domain, with a 2D cross-section $\Omega$, and height, $h$. Then the 3D LdG energy is given by :
\begin{equation}
    \label{eq:3Denergy}
    I[\Qvec]: = \int_{\Omega \times [0,h]} \frac{L}{2}|\nabla \Qvec |^2 + F_b(\Qvec)~dV
\end{equation}
where $F_b(\Qvec): = \frac{A}{2}\textrm{tr}\Qvec^2 - \frac{B}{3} \textrm{tr}\Qvec^3 + \frac{C}{4} (\textrm{tr}\Qvec^2)^2$, $A$ is a re-scaled temperature, $B$ and $C$ are positive material-dependent constants \cite{mottram2014introduction}. The constant $L>0$ is a material-dependent elastic constant, with units of Newtons. Typical values of $L$ are $10^{-12}$ Newtons \cite{priestly2012introduction}. We have adopted the Dirichlet elastic energy density, $w(\nabla \Qvec) \propto |\nabla \Qvec|^2 = \sum_{i,j,k=1}^{3} Q_{ij,k}^2$, where $Q_{ij,k} = \frac{\partial Q_{ij}}{\partial x_k}$, based on the assumption that all elastic deformations e.g., splay, twist and bend deformations are energetically degenerate. The physically observable stable equilibria are modelled by minimisers of \eqref{eq:3Denergy} in an appropriately defined admissible space.

For thin 3D systems, for which $h$ is much smaller than the dimensions of $\Omega$, it suffices to work with the reduced Landau--de Gennes (rLdG) model, based on the assumption that the nematic director is in the cross-section plane, and structural details are invariant across the height of the system \cite{golovatyreduced}. 
The rLdG model has been successfully applied for capturing the qualitative properties of physically relevant solutions and for probing into defect cores \cite{brodin2010melting,gupta2005texture,Mu2008Self,Igor2006Two}. In the rLdG model, the nematic state in the 2D cross-section/2D domain is described by a reduced order parameter: a symmetric, traceless $2\times 2$ matrix, $\mathbf{P}$, as given below
\begin{equation} \label{eq:P}
    \mathbf{P} =
    \left(\begin{tabular}{cc}
    $P_{11}$&$P_{12}$\\
    $P_{12}$&$-P_{11}$\\
\end{tabular}\right).
\nonumber
\end{equation}
We work at a special temperature, $A = -B^2/3C$, where $B$ and $C$ are as before \cite{de1993physics}, and with the following rLdG free energy:
\begin{equation}
        F[\mathbf{P}]: = \int_{\Omega}\frac{L}{2}|\nabla \mathbf{P}|^2+f_b(\Pvec) \mathrm{d A}.
        \label{rLdG_energy}
    \end{equation}
where $\Omega$ is the 2D cross-section of our 3D domain. In the remainder of this manuscript, $\Omega$ is the complement of a regular 2D polygon with $K$ edges, $E^C_K$, in $\mathbb{R}^2$. Here, we have employed the Dirichlet elastic energy density as in \eqref{eq:3Denergy} and the bulk energy density as
$$
f_b(\Pvec) = -\frac{B^2}{4C}\textrm{tr}\Pvec^2 + \frac{C}{4}(\textrm{tr}\Pvec^2)^2.$$
We choose this temperature, partly for comparison with previous work in \cite{han2020pol} and \cite{fangmajumdarzhang2020}, and partly because for this special temperature,  the critical points of the rLdG model exist as critical points of the full 3D LdG free energy in 3D settings, for suitably defined boundary conditions. This is generally not true for arbitrary $A<0$; see \cite{canevariharrismajumdarwang2020} for more details. More precisely, at $A=-B^2/3C$, given a critical point $\Pvec_c$ of \eqref{rLdG_energy}, there exists a critical points $\Qvec_c$ of \eqref{eq:3Denergy} such that 

$$ \Qvec_c = \Pvec_3 - \frac{B}{6C}\left(2\zhat\otimes\zhat - \xhat\otimes\xhat - \yhat\otimes\yhat \right), $$
where $\xhat, \yhat. \zhat$ are coordinate unit-vectors.
The matrix, $\Pvec_3$ is a $3\times 3$ symmetric traceless matrix, such that $ (\Pvec_3)_{ij} = (\Pvec_{c})_{ij}$ for $i,j=1,2$ and all remaining matrix entries are set to zero.

The energy, \eqref{rLdG_energy} is non-dimensionalised with $\overline{x} = x/\bar{\lambda}$, where $\bar{\lambda}$ is the edge-length of the polygonal hole. 
\begin{equation}
        F[\mathbf{P}]: = \int_{E_K^C}\frac{1}{2}|\bar{\nabla} \mathbf{P}|^2+\frac{\bar{\lambda}^2}{L}f_b(\Pvec) \mathrm{d A}.
        \label{p_energy}
    \end{equation}
The working domain is now $E_K^C$, the complement of a 2D re-scaled polygon $E_K$ with $K$ edges of unit length, centered at the origin, with vertices
\begin{equation}
    w_k =  \left(\cos\left(2\pi \left(k-1\right)/K\right),\sin\left(2\pi \left(k-1\right)/K\right)\right),\ k = 1,...,K.\nonumber
\end{equation}
We label the edges counterclockwise as $C_1, ..., C_K$, starting from $\left(1,0\right)$. See figure \ref{fig:domain}. In the following, we drop the bar over $\nabla$ for brevity.

\begin{figure}
\centering
        \includegraphics[width=0.5\columnwidth]{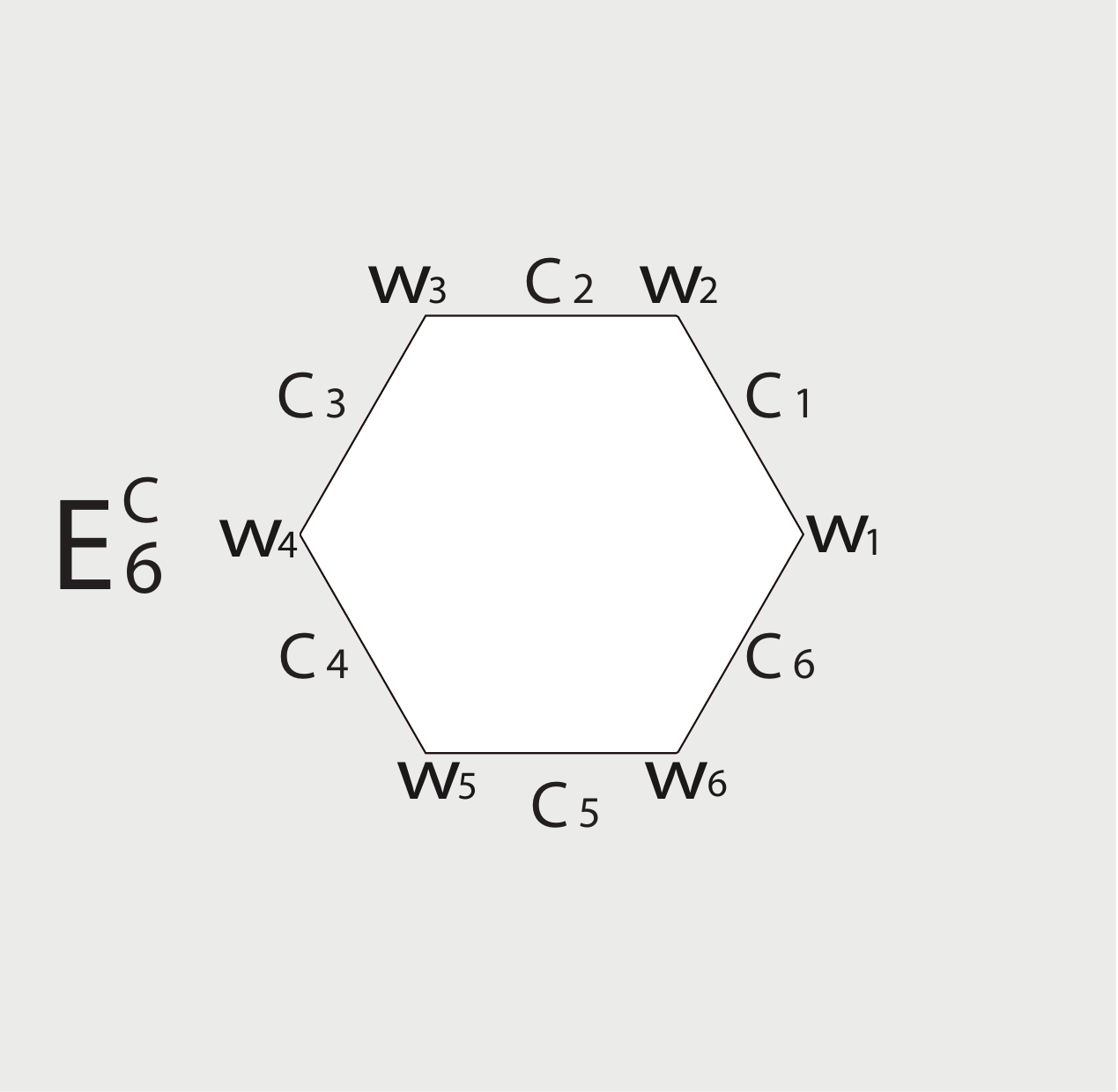}
        \caption{The normalized domain $E_K^C$, with $K = 6$, as an illustrative example.}
        \label{fig:domain}
\end{figure}

We can also write $\mathbf{P}$ in terms of an order parameter, $s$, and an angle $\gamma$ as shown below -
\begin{equation}
    \mathbf{P} = 2s\left(\mathbf{n}\otimes\mathbf{n}-\frac{1}{2}\mathbf{I}_2\right),
    \label{P}
\end{equation}
where $\mathbf{n} = \left(\cos\gamma,\sin\gamma\right)^T$ is the nematic director in the plane, and $\mathbf{I}_2$ is the $2\times 2$ identity matrix,
so that
\[P_{11} = s\cos\left(2\gamma\right),\ P_{12} = s\sin\left(2\gamma\right).\]
The defect set is
simply identified with the nodal set of $s$, also see \cite{han2020pol}, interpreted as the set of no planar order in $E_K^C$.

We impose homeotropic boundary conditions on $\partial E_K$, which requires $\mathbf{n}$ in \eqref{P} to be homeotropic/normal to the edges of $E_K$. However, there is a necessary mismatch at the corners/vertices. We impose a continuous Dirichlet boundary condition, $\mathbf{P} = \mathbf{P}_b$, on $\partial E_K$ as defined below:
\begin{equation}
    \begin{aligned}
    &P_{11b}\left(w\right) = \alpha_k = \begin{cases}
    &\frac{B}{2C}\left((1-\frac{a}{2})\cos\left(\frac{\left(2k-1\right)2\pi}{K}\right) + \frac{a}{2}\cos\left(\frac{\left(2k+1\right)2\pi}{K}\right)\right), ||w-w_k||\geq ||w-w_{k+1}||,\\
    &\frac{B}{2C}\left((1-\frac{a}{2})\cos\left(\frac{\left(2k-1\right)2\pi}{K}\right) + \frac{a}{2}\cos\left(\frac{\left(2k-3\right)2\pi}{K}\right)\right), ||w-w_k||\leq ||w-w_{k+1}||,
    \end{cases}\\
    &P_{12b}\left(w\right) = \beta_k = \begin{cases}
    &\frac{B}{2C}\left((1-\frac{a}{2})\sin\left(\frac{\left(2k-1\right)2\pi}{K}\right) + \frac{a}{2}\sin\left(\frac{\left(2k+1\right)2\pi}{K}\right)\right), ||w-w_k||\geq ||w-w_{k+1}||,\\
    &\frac{B}{2C}\left((1-\frac{a}{2})\sin\left(\frac{\left(2k-1\right)2\pi}{K}\right) + \frac{a}{2}\sin\left(\frac{\left(2k-3\right)2\pi}{K}\right)\right), ||w-w_k||\leq ||w-w_{k+1}||,
    \end{cases}
    \end{aligned}
    \label{Pb}
\end{equation}
 where $a(w,\sigma)\in[0,1]$, $w\in C_k$, $k= 1,\cdots, K$ is an interpolation function, which is continuous and strictly monotonic about $||w-(w_k+w_{k+1})/2||$, increasing from  $0$  at the centre of $C_k$, to $1$ at the two vertices of $C_k$. Further, we take $0<\sigma\leq 1/2$ and $a(w,\sigma)\to 0$ as $\sigma\to 0$.
At the vertices $w = w_k$, we set $\mathbf{P}_b$ to be equal to the average of the two constant values on the two intersecting edges, and at the edge mid-point i.e., $w = (w_k+w_{k+1})/2$, we have strictly homeotropic conditions. 
As $\sigma\to 0$, the Dirichlet boundary conditions are piece-wise constant and compatible with strict homeotropic conditions
\begin{gather}\label{Pb_constant}
P_{11b}(w) = \hat{\alpha}_k = \frac{B}{2C}\cos\left(\frac{\left(2k-1\right)2\pi}{K}\right), \qquad P_{12b}(w) = \hat{\beta}_k = \frac{B}{2C}\sin\left(\frac{\left(2k-1\right)2\pi}{K}\right),\ w\in C_k.
\end{gather}

Further, the domain $E_K^C$ is unbounded and we impose uniform/constant boundary conditions at infinity as given below:
\begin{equation}\label{constraint_infty}
\lim_{|w|\to\infty} \mathbf{P}= \mathbf{P}^*:= \frac{B}{C}(\mathbf{n}^*\otimes \mathbf{n}^* - \frac{1}{2}\mathbf{I}_2),
\end{equation}
where $\mathbf{n}^* = (\cos(\gamma^*),\sin(\gamma^*))$ is the constant nematic director at infinity. To avoid confusion, we reiterate that $\gamma$ is the director angle and $\gamma^*$ is associated with the boundary condition at infinity.

The corresponding Euler-Lagrange equations are
\begin{equation}
    \begin{aligned}
        \Delta P_{11} &= \frac{2C\bar{\lambda}^2}{L}\left(P_{11}^2+P_{12}^2-\frac{B^2}{4C^2}\right)P_{11},\\
        \Delta P_{12} &= \frac{2C\bar{\lambda}^2}{L}\left(P_{11}^2+P_{12}^2-\frac{B^2}{4C^2}\right)P_{12}.
\end{aligned}
    \label{Euler_Lagrange}
\end{equation}
For the purposes of this paper, the crucial dimensionless parameter is
\[
\lambda^2 = \frac{ 2 C}{L} \bar{\lambda}^2
\]
where the nematic correlation length, $\xi_n \propto \sqrt{\frac{L}{C}}$ at the fixed temperature $A = -\frac{B^2}{3C}$. Recall that $L$ has units of \emph{Newton (N)} and $C$ has units of \emph{$N m^{-2}$}, so that $\xi_n$ has the units of length.

The admissible space is
\begin{align}
&\mathcal{H}_{\infty}:=\mathbf{P}^* + \mathcal{H},\\
&\mathcal{H}:=\left\{\mathbf{H}\in H^1_{loc}(E_K^C;\mathcal{S}_0):\int_{E_K^C}|\nabla \mathbf{H}|^2 + \int_{E_K^C}\frac{|\mathbf{H}|^2}{|w|^2}<\infty\right\}
\end{align}
where
\begin{equation}\nonumber
    \mathcal{S}_0:=\{\mathbf{P}\in M_2(\mathbb{R}): P_{ij} = P_{ji},tr(\mathbf{P}) = 0\}.
\end{equation}
The free energy \eqref{p_energy} is not always finite
, since the potential term $f_b\geq 0$ may very well not be integrable in $E_K^C$. However, since $f_b(\mathbf{P}^*) = 0$,  we may find a compactly supported $\mathbf{H}\in\mathcal{H}$, for which $\mathbf{P} = \mathbf{P}^* + \mathbf{H}$ satisfies the boundary condition \eqref{Pb}.
The existence of solution of \eqref{Euler_Lagrange} in $\mathcal{H}_{\infty}$, has been proven in Proposition 3 of \cite{Bronsard2016minimizers}. 

We study two distinguished limits in what follows---the $\lambda\to0$ limit which is relevant for polygon holes $E_K$ with edge length comparable to $\xi_n$ which is typically on the nanometre scale, and the $\lambda\to\infty$ limit, which is the macroscopic limit relevant for micron-scale or larger polygonal holes. In the following sections, we study the limiting problems and the limiting minimiser profiles, their defect sets and multistability in the $\lambda \to \infty$ limit.

\section{The $\lambda\to0$ limit}
\label{sec:lambda0}
In Theorem 1 of \cite{Bronsard2016minimizers}, as $\lambda\to 0$, the solution of \eqref{Euler_Lagrange} converges to the unique solution of \eqref{zero_euler} below, with Dirichlet boundary conditions. 
\begin{equation}
\begin{aligned}
    &\Delta P_{11}^0 = 0,\ \Delta P_{12}^0 = 0, on\ E_K^C,\\
    &P_{11}^0 = P_{11b},\ P_{12}^0 = P_{12b},\ on\ \partial E_K,\\
    &\mathbf{P} = \mathbf{P}^*,\ |x|\to\infty.
\end{aligned}
    \label{zero_euler}
\end{equation}
In other words, the limiting problem is a boundary-value problem for the Laplace equation for $\Pvec$ on $E^C_K$, in the $\lambda \to 0$ limit. This problem is explicitly solvable and in the following sections, we exploit the symmetries of the Laplace equation, the symmetries of the polygon and boundary conditions to illustrate how the limiting profile depends on $K$ - the number of polygon edges, and $\gamma^*$ - the director angle at infinity. In fact, these two parameters tune the existence, location and dimensionality of defects in this limit, amenable to experimental verification in due course.

\subsection{Defect patterns outside a 2D disc}
As an illustrative example, we first consider the limiting problem \eqref{zero_euler} on the complement of a disc. As $K\to\infty$, the domain, $E_K^C$, converges to the exterior of a disc, $D^C$.
The conformal mapping from unit disc, $D$, to exterior of disc, $D^C$, is given by
$$
w = f(z) = \frac{1}{z}.
$$
Under the mapping $f^{-1}:D^C\to D$, the limiting problem for $\lambda = 0$ is given by:
\begin{align}
&\Delta p_{11} = 0, \ \Delta p_{12} = 0,\ z\in D \nonumber \\
&p_{11} = \frac{B}{2C}\cos(-2\theta),\ p_{12} = \frac{B}{2C}\sin(-2\theta),\ z\in\partial D \nonumber \\
&p_{11} = \frac{B}{2C}\cos(2\gamma^*),\ p_{12} = \frac{B}{2C}\sin(2\gamma^*),\ z=0
\end{align}
where $z = r e^{i\theta}$, $\theta$ is the azimuthal angle and $r$ is the radius.
The corresponding solution  is 
\begin{align}\label{disc_zero}
p_{11}(re^{i\theta},\gamma^*) &= \frac{B}{2C}\left(r^2\cos(-2\theta) + \lim_{\epsilon\to 0}\frac{\cos(2\gamma^*)ln r}{ln \epsilon}\right),\nonumber \\
p_{12}(re^{i\theta},\gamma^*) &= \frac{B}{2C}\left(r^2\sin(-2\theta)  + \lim_{\epsilon\to 0}\frac{\sin(2\gamma^*)ln r}{ln \epsilon}\right).
\end{align}
This solution has the rotational symmetry property 
\begin{multline}
(p_{11},p_{12})(re^{i\theta},\gamma^*) = (p_{11}(re^{i\theta+i\gamma^*},0)\cos2\gamma^* - p_{12}(re^{i\theta+i\gamma^*},0)\sin 2\gamma^*,\\
p_{12}(re^{i\theta+i\gamma^*},0)\cos 2\gamma^* + p_{11}(re^{i\theta+i\gamma^*},0)\sin2\gamma^*),   
\end{multline}
so that it suffices to assume $\gamma^* = 0$. 
With $\gamma^* = 0$, we can check that $p_{11}=p_{12}=0$ at exactly two points, located at $\theta=\frac{\pi}{2}$ and $\theta = \frac{3\pi}{2}$ respectively. 
The corresponding limiting solution on $D^C$ is: 
\begin{equation}\label{2d}
\mathbf{P}(w, \gamma^*) = \mathbf{p}(f^{-1}(w), \gamma^*) = s_+\left(\frac{1}{\rho^2}(\mathbf{e}_{\rho}\otimes \mathbf{e}_{\rho}-\mathbf{I}_2/2 )+\lim_{\epsilon\to 0}\frac{ln \rho}{ln 1/\epsilon}(\mathbf{n}^*\otimes \mathbf{n}^* -\mathbf{I}_2/2)\right)
\end{equation}
where $w = \rho e^{i\psi}$,$\mathbf{e}_{\rho}= (\cos\psi,\sin\psi)$ is the unit radial vector and the constraint at infinity is $\mathbf{n}^* = \left(\cos\gamma^*,\sin\gamma^*\right)$.


This method can be easily generalised to piecewise constant boundary conditions on segments of $\partial D$, relevant for solving the limiting problem (\ref{zero_euler}) on $E^C_K$, as we shall see in the Section 3.3.
Consider the following boundary-value problem on the unit disc $D$,
\begin{align}\label{eq:u}
&\Delta u = 0,\ z\in D, \nonumber\\
&u = u_k,\ on\ z\in D_k,\ k = 1,\cdots, K,\nonumber \\
&u = u_0,\ on\ z = 0.
\end{align}
where 
\begin{equation}\label{D_k}
D_k = \{e^{i\theta},\theta\in(-2\pi k/K, -2\pi k/K + 2\pi/K)\},\ k = 1,...,K,
\end{equation}
are the segments of $\partial D$.
The solution, $u$, can be written as, $u = u_a + u_b$, where $u_a$ and $u_b$ are defined by the following boundary-value problems:
\begin{align}\label{eq:ua}
&\Delta u_a = 0,\ z\in D,\nonumber \\
&u_a = u_k,\ on\ z\in D_k,\ k = 1,\cdots, K, \nonumber \\
&u_a = \frac{1}{2\pi}\sum_{k = 1}^K \int_{2\pi (K-k)/K}^{2\pi(K-k+1)/K}u_k d\phi\, on\  z=0,
\end{align} and
\begin{align}\label{eq:ub}
&\Delta u_b = 0,\ z\in D, \nonumber \\
&u_b = 0,\ on\ z\in D_k,\ k = 1,\cdots, K, \nonumber\\
&u_b = u_0-\frac{1}{2\pi}\sum_{k = 1}^K \int_{2\pi (K-k)/K}^{2\pi(K-k+1)/K}u_k d\phi\, on\  z=0.
\end{align} 
Using the Poisson integral, and standard separation of variables method for the Laplace equation on an annulus, with the radius of inner ring (denoted by $\epsilon$) approaching zero, the solution of \eqref{eq:u} is given by
\begin{equation}\label{u_solution}
u(re^{i\theta}) = \frac{1}{2\pi}\sum_{k = 1}^K\int_{2\pi (K-k)/K}^{2\pi(K-k+1)/K}u_k\frac{1-r^2}{1-2r\cos(\phi-\theta)+r^2}d\phi + \lim_{\epsilon\to 0}\frac{(u_0- \frac{1}{2\pi}\sum_{k = 1}^K \int_{2\pi (K-k)/K}^{2\pi(K-k+1)/K}u_k d\phi)ln r}{ln \epsilon}.
\end{equation}

\subsection{Schwarz--Christoffel mapping}
The  example of $D^C$ gives useful insights into the computation of the limiting profile for the exterior of a regular polygon, $E^C_K$.
The conformal mapping from $D$ to $D^C$ is straightforward to construct. The analogous conformal mapping from $D$ to $E_K^C$ is the following Schwarz--Christoffel mapping:
\begin{figure}
\centering
        \includegraphics[width=0.5\columnwidth]{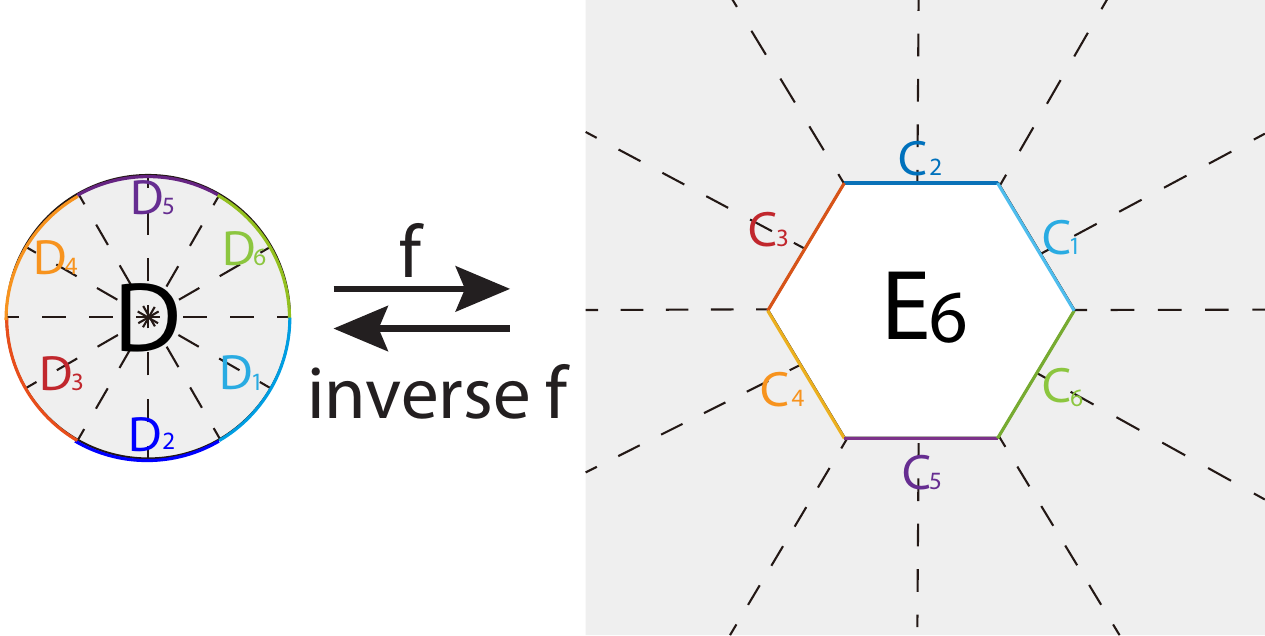}
        \caption{Schwarz-Christoffel mapping from unit disc to exterior of hexagon.}
        \label{fig:SC}
\end{figure}
The mapping from $D$ to $E_K^C$ (exterior of a regular polygon with $K$ sides of unit length since we have re-scaled the length variable) is defined by \cite{driscoll2002schwarz}
\begin{equation}
w = f(z) = A - C\int^z x^{-2}\prod_{k = 1}^K\left(1-\frac{x}{w_k} \right)^{1-\alpha_k} dx,\ \forall z\in D.
\end{equation}
Here $\alpha_k\pi$ is the exterior angle of $E_K^C$ (interior angle of polygon $E_K$) and $\omega_k$ are the polygon vertices.
In particular, for the mapping from the unit disc to $E^C_K$, $\alpha_k = 1-2/K$ $k = 1,\cdots,K$, when the first vertex is located at $w_1 = (1,0)$ and the Schwarz--Christoffel mapping is uniquely given by:
\begin{equation}\label{eq:SC}
w = f(z) = C(K)\int^z x^{-2}\left (1-x^K \right)^{2/K}dx,
\end{equation}
see Fig. \ref{fig:SC}.
The leading term of $f$ in \eqref{eq:SC} is $-C(K)z^{-1}$ and hence, as $z\to 0$, $f(z) \to \infty$ and $f$ is single-valued near the origin \cite{driscoll2002schwarz}.
The pre-factor, $C(K)$ is real and $|C(K)|$ is the capacity or transfinite diameter of the region $E_K^C$ \cite{driscoll2002schwarz}.
One can check that $f$ maps the circle, $\partial D$, onto the polygon boundary, $\partial E_K = f(\partial D)$ and the segments of $\partial D$ to the corresponding segments of $\partial E_K^C$, i.e., 
\begin{equation}
f(D_k) = C_k,
\end{equation}
where $D_k$, $k = 1,\cdots,K$ is defined in \eqref{D_k}, and 
 \begin{equation}
     f(0) = \infty
 \end{equation}.

The mapping, $f$, has the following properties
\begin{align}
f(\overline{z}) &=  C(K)\int^{\overline{z}} x^{-2}(1-x^K)^{2/K}dx = C(K)\int^z \overline{x}^{-2}(1-\overline{x}^K)^{2/K}dx = \overline{f(z)},\label{f_reflection}\\
f(ze^{2\pi n i/K}) &= C(K)\int^{ze^{2\pi n i/K}} x^{-2}(1-x^K)^{2/K}dx\\ 
&= C(K)\int^{z} e^{-4\pi n i/K}x^{-2}(1-x^K)^{2/K} e^{2\pi n i/K}dx = e^{-2\pi n i/K} f(z),\label{f_rotation}
\end{align} which are useful when studying the properties of solutions of (\ref{zero_euler}) on $E^C_K$.

\subsection{The limiting problem for $E^C_K$ in terms of a boundary-value problem on $D$}
Under the Schwarz-Christoffel mapping $f^{-1}:E_K^C\to D$, the limiting problem \eqref{zero_euler} can be equivalently written in terms of a rLdG tensor, $\mathbf{p}$ defined on $D$, as shown below:
\begin{equation}
\begin{aligned}
    &\Delta p_{11}^0 = 0,\ \Delta p_{12}^0 = 0, on\ D,\\
    &p_{11}^0 = p_{11b},\ p_{12}^0 = p_{12b},\ on\ D_k,\\
    &p_{11}(0,0) = \frac{B}{2C}\cos(2\gamma^*),\ p_{12}(0,0) = \frac{B}{2C}\sin(2\gamma^*),\\
\end{aligned}
\label{eq:map_zero}
\end{equation}
The Dirichlet boundary conditions (\ref{Pb}) translate to boundary conditions, on $D_k$, $k = 1,\cdots, K$ defined in \eqref{D_k}, which are segments of $\partial D$; see below:
\begin{equation}
    \begin{aligned}
    &p_{11b} = \overline{\alpha}_k = \begin{cases}
    &\frac{B}{2C}\left((1-\frac{\overline{a}}{2})\cos\left(\frac{\left(2k-1\right)2\pi}{K}\right) + \frac{\overline{a}}{2}\cos\left(\frac{\left(2k+1\right)2\pi}{K}\right) \right),\ \theta\in[\frac{2\pi(K-k)}{K},\frac{2\pi(K-k+1)}{K}-\frac{\pi}{K}],\\
    &\frac{B}{2C}\left((1-\frac{\overline{a}}{2})\cos\left(\frac{\left(2k-1\right)2\pi}{K}\right) + \frac{\overline{a}}{2}\cos\left(\frac{\left(2k-3\right)2\pi}{K}\right)\right),\ \theta\in[\frac{2\pi(K-k+1)}{K}-\frac{\pi}{K},\frac{2\pi(K-k+1)}{K}],
    \end{cases}\\
    &p_{12b} = \overline{\beta}_k = \begin{cases}
    &\frac{B}{2C}\left((1-\frac{\overline{a}}{2})\sin\left(\frac{\left(2k-1\right)2\pi}{K}\right) + \frac{\overline{a}}{2}\sin\left(\frac{\left(2k+1\right)2\pi}{K}\right)\right),\ \theta\in[\frac{2\pi(K-k)}{K},\frac{2\pi(K-k+1)}{K}-\frac{\pi}{K}],\\
    &\frac{B}{2C}\left((1-\frac{\overline{a}}{2})\sin\left(\frac{\left(2k-1\right)2\pi}{K}\right) + \frac{\overline{a}}{2}\sin\left(\frac{\left(2k-3\right)2\pi}{K}\right)\right),\ \theta\in[\frac{2\pi(K-k+1)}{K}-\frac{\pi}{K},\frac{2\pi(K-k+1)}{K}],
    \end{cases}
    \end{aligned}
    \label{pb}
\end{equation}
where the interpolation function, $\overline{a}(z,\sigma)$, $\forall z\in \partial D$ satisfies $\overline{a}(e^{i(2K-2k+1)\pi/K},\sigma) = a((w_k+w_{k+1})/2,\sigma) = 0$ and $\overline{a}(e^{i2(K-k)\pi/K},\sigma) = a(w_{k+1},\sigma) = 1$. We assume $\overline{a}\in C^2$ and $\frac{\partial^2\overline{a}}{\partial \theta^2}\geq 0$.
As $\sigma\to 0$, $\overline{a}\to 0$, the Dirichlet boundary conditions approach $p_{11b} = \hat{\alpha}_k$, $p_{12b} = \hat{\beta}_k$ on $D_k$ uniformly, where $\hat{\alpha}_k$ and $\hat{\beta}_k$ are given in \eqref{Pb_constant}. 

For convenience, we extend the definition of $\overline{\alpha}_k$, $\overline{\beta}_k$, $k = 1,\cdots,K$, to $k\in\mathbb{Z}$ and use the periodicity of $\tan$, $\cos$ and $\sin$ to define
\begin{gather}
\overline{\alpha}_{k+nK} = \overline{\alpha}_k,\overline{\beta}_{k+nK} = \overline{\beta}_k,\ n\in\mathbb{Z}.
\end{gather}

One can check the following relations between $\overline{\alpha}_k$ and $\overline{\alpha}_{n+k}$($\overline{\alpha}_{K-k+1}$), $\overline{\beta}_k$ and $\overline{\beta}_{n+k}$($\overline{\beta}_{K-k+1}$):
\begin{align}
\overline{\alpha}_{n+k} &= \cos(4\pi n/K)\overline{\alpha}_k-\sin(4\pi n/K)\overline{\beta}_k,\label{alpha_rotation}\\
\overline{\beta}_{n+k} &= \sin(4\pi n/K)\overline{\alpha}_k+\cos(4\pi n/K)\overline{\beta}_k,\label{beta_rotation}\\
\overline{\alpha}_{K-k+1} &= \overline{\alpha}_k,\label{alpha_reflection}\\
\overline{\beta}_{K-k+1} &= -\overline{\beta}_k.\label{beta_reflection}
\end{align}
The constants $\hat{\alpha}_k$ and $\hat{\beta}_k$ have similar properties. 
From \eqref{beta_reflection}, we have 
\begin{align}
\int_{0}^{2\pi} p_{12b} d\theta &= \sum_{k = 1}^K \overline{\beta}_k = \frac{1}{2}\sum_{k = 1}^K\overline{\beta}_k+\frac{1}{2}\sum_{k = 1}^{K}\overline{\beta}_{K-k+1} = 0.\label{beta_zero}
\end{align}
From \eqref{beta_rotation} and \eqref{beta_zero}, we have
\begin{align}\label{p11_center_zero}
\int_{0}^{2\pi} p_{11b} d\theta &= \sum_{k = 1}^K \overline{\alpha}_k  = \sum_{k = 1}^K \frac{\overline{\beta}_{k+1}-cos(4\pi/K)\overline{\beta}_k}{\sin(4\pi/K)} \nonumber\\
&= \frac{1}{\sin(4\pi/K)}\sum_{k = 1}^K \overline{\beta}_{k+1}-\frac{1}{tan(4\pi/K)}\sum_{k = 1}^K \overline{\beta}_{k} = 0,\ for\ K\neq4.
\end{align}
Additionally, for $K =4$, $\overline{\alpha}_k = 0$, $k = 1,\cdots,4$, i.e., $\int_{0}^{2\pi} p_{11b} d\theta = 0$.

Hence, following the solution of \eqref{eq:u} with Dirichlet boundary conditions on $\partial D$, and constraint at origin in \eqref{u_solution}, and the results in \eqref{beta_zero}--\eqref{p11_center_zero}, the solution of \eqref{eq:map_zero} is given by
\begin{align}
p_{11}(r,\theta,\gamma^*) 
& =  \frac{1}{2\pi}\sum_{k = 1}^K\int_{2\pi (K-k)/K}^{2\pi(K-k+1)/K}\overline{\alpha}_{k}\frac{1-r^2}{1-2r\cos(\phi-\theta)+r^2}d\phi + \lim_{\epsilon\to 0}\frac{B}{2C}\frac{\cos(2\gamma^*)ln r}{ln \epsilon}\nonumber \\
& =  \frac{1}{2\pi}\sum_{k = K}^1\int_{2\pi (k-1)/K}^{2\pi k/K}\overline{\alpha}_{K-k+1}\frac{1-r^2}{1-2r\cos(\phi-\theta)+r^2}d\phi + \lim_{\epsilon\to 0}\frac{B}{2C}\frac{\cos(2\gamma^*)ln r}{ln \epsilon}\nonumber \\
& =  \frac{1}{2\pi}\sum_{k = 1}^K\int_{2\pi(k-1)/K}^{2\pi k/K}\overline{\alpha}_k\frac{1-r^2}{1-2r\cos(\phi-\theta)+r^2}d\phi + \lim_{\epsilon\to 0}\frac{B}{2C}\frac{\cos(2\gamma^*)ln r}{ln \epsilon}\label{zero_solution_p11}\\
p_{12}(r,\theta,\gamma^*) & =  \frac{1}{2\pi}\sum_{k = 1}^K\int_{2\pi (K-k)/K}^{2\pi(K-k+1)/K}\overline{\beta}_{k}\frac{1-r^2}{1-2r\cos(\phi-\theta)+r^2}d\phi + \lim_{\epsilon\to 0}\frac{B}{2C}\frac{\sin(2\gamma^*)ln r}{ln \epsilon}\nonumber\\
& =  \frac{1}{2\pi}\sum_{k = 1}^K\int_{2\pi (k-1)/K}^{2\pi k/K}\overline{\beta}_{K-k+1}\frac{1-r^2}{1-2r\cos(\phi-\theta)+r^2}d\phi + \lim_{\epsilon\to 0}\frac{B}{2C}\frac{\sin(2\gamma^*)ln r}{ln \epsilon}\nonumber\\
& =  -\frac{1}{2\pi}\sum_{k = 1}^K\int_{2\pi(k-1)/K}^{2\pi k/K}\overline{\beta}_k\frac{1-r^2}{1-2r\cos(\phi-\theta)+r^2}d\phi + \lim_{\epsilon\to 0}\frac{B}{2C}\frac{\sin(2\gamma^*)ln r}{ln \epsilon}\label{zero_solution_p12}.
\end{align}
In the above, we use (\ref{alpha_reflection}) and (\ref{beta_reflection}), along with standard changes of variables i.e. $s - 1 = K - k$, and swap the summation variable from $s$ to $k$ etc. The following proposition is crucial for restricting $\gamma^*$ to a specified range, for a given $K$, using rotation and reflection symmetries, the proof of which deferred to an appendix, so as not to detract from the main details.
The following corollaries describe the symmetry properties of the limiting solution $\mathbf{P}$ on $E_K^C$, for any $K$.

\begin{proposition}\label{gamma_infinity_restriction}
We can restrict $\gamma^*\in[0,\frac{\pi}{K}]$, since there are rotation relations between\\ $(p_{11}, p_{12})\vert_{(re^{i\theta+2\pi ki/K},\gamma^*-\frac{2\pi k}{K})}$, $k = 1,\cdots,K$, and $(p_{11}, p_{12})\vert_{(re^{i\theta},\gamma^*)}$, and reflection relations between $(p_{11}, p_{12})\vert_{(re^{i\theta},\gamma^*)}$ and $(p_{11}, p_{12})\vert_{(re^{-i\theta},-\gamma^*)}$.
\end{proposition}
\begin{proof} See Appendix.
\end{proof}
\begin{corollary}\label{remark1}
 If $K$ is even, using \eqref{p11_rotation} and \eqref{p12_rotation} with $n = K/2$, we have the symmetry property $\mathbf{p}(re^{i\theta+i\pi},\gamma^*) = \mathbf{p}(re^{i\theta},\gamma^*)$. Due to the property of the SC mapping $f$ in \eqref{f_rotation}, $f(re^{i\theta+i\pi}) = e^{-i\pi}f(re^{i\theta}) = e^{i\pi}f(re^{i\theta})$, the corresponding $\mathbf{P}$, defined on $E_K^C$, has the symmetry property: $\mathbf{P}(\rho e^{i\psi+i\pi},\gamma^*) = \mathbf{P}(\rho  e^{i\psi},\gamma^*)$.
\end{corollary}

\begin{corollary}\label{0_symmetry}
From \eqref{p11_reflection} and \eqref{p12_reflection} with $\gamma^* = 0$, we have
\begin{gather}
p_{11}(re^{-i\theta},0) = p_{11}(re^{i\theta},0),\qquad p_{12}(re^{-i\theta},0) = -p_{12}(re^{i\theta},0).
\end{gather}
The SC mapping in \eqref{eq:SC} preserves reflection symmetry, $f(re^{-i\theta}) = \overline{f(re^{i\theta})}$.
We have $(P_{11},P_{12})(\rho e^{\psi i},0) = (p_{11},p_{12})(f^{-1}(\rho e^{\psi i}),0)$, for $w=\rho e^{\psi i}\in E_K^C$, $(P_{11},P_{12})$ has reflection symmetry about $\psi = 0$, i.e.,
\begin{gather}
P_{11}(\rho e^{-i\psi},0) = P_{11}(\rho e^{i\psi},0),\qquad P_{12}(\rho e^{-i\psi},0) = -P_{12}(\rho e^{i\psi},0).
\end{gather}
\end{corollary}

\begin{corollary}\label{piK_symmetry}
With $\gamma^* = \pi/K$, $(P_{11},P_{12})$ has reflection symmetry about $\psi = \pi/K$, i.e.,
\begin{align}
P_{11}(\rho e^{\pi/K i-\psi i},\pi/K) &= P_{11}(\rho e^{\pi/K i+ \psi i},\pi/K)\cos(4\pi/K) + P_{12}(\rho e^{\pi/K i+ \psi i},\pi/K)\sin(4\pi/K),\\
P_{12}(\rho e^{\pi/K i-\psi i},\pi/K) &= -P_{12}(\rho e^{\pi/K i+ \psi i},\pi/K)\cos(4\pi/K) + P_{11}(\rho e^{\pi/K i + \psi i},\pi/K)\sin(4\pi/K).
\end{align}
\end{corollary}
\begin{proof} See Appendix.
\end{proof}

In the next sub-sections, we apply these results to $K=3$, $K=4$ and generic $E^C_K$ with $K>4$. These specific examples demonstrate the interplay between $K$ and $\gamma^*$, and how this can be exploited to tailor defect sets in reduced 2D problems.

\subsubsection{The limiting profile for $E^C_3$ - exterior of an equilateral triangle}

We first consider the solution of (\ref{zero_euler}) on $E^C_3$- the complement of a regular, re-scaled equilateral triangle with homeotropic boundary conditions. Recall the SC mapping from the unit disc $D$ to $E^C_3$ given by $f(z) = f(re^{i\theta }) = \rho e^{i \psi} = w$.

From Corollary \ref{0_symmetry}, for any $K$, with $\gamma^* = 0$, we have $p_{12}(r,0) \equiv 0$ on $\theta=0$ or $\pi$. The boundary condition in \eqref{pb}, $p_{12b} = \beta_k\leq 0$ with $K=3$ on $\partial D\cap \{\theta\in[0,\pi]\}$. From the maximum principle for the Laplace equation on $D\cap\{\theta\in[0,\pi]\}$, $p_{12}= 0$ if and only if $\theta = 0$ or $\pi$. Analogously, on $D\cap\{\theta\in[\pi,2\pi]\}$, $p_{12}=0$ if and only if $\theta = \pi$ or $2\pi$. In conclusion, defects can only appear on the diameter $z = r$, $r\in(-1,1)$. 
Using \eqref{p11_center_zero}, we obtain $p_{11}(0,\gamma^*)-\frac{B}{2C}\cos(2\gamma^*)c(r) = \frac{1}{2\pi}\int_0^{2\pi}p_{11b}d\theta = 0$, where $c(r) = \lim_{\epsilon\to 0}\frac{ln r}{ln \epsilon}$.
For $K = 3$, with $\gamma^* = 0$, as in Lemma 4.5 in \cite{canevari2017order}, we show that as $r$ increases from $0$ to $1$, $p_{11}(r,0)-\frac{B}{2C}c(r)$ monotonically decreases from $0$ to $\overline{\alpha}_1 = \frac{B}{2C}\left((1-\overline{a}/2)\cos(2\pi/3) + \overline{a}/2\cos(-2\pi/3)\right) = \frac{B}{2C}\cos(2\pi/3) = -\frac{B}{4C}$. The term, $\frac{B}{2C}c(r)$, monotonically decreases from $\frac{B}{2C}$ to $0$. Hence, $p_{11}(r,0)$ monotonically decreases from $B/2C$ to $-B/4C$. There exists a $r^*$ such that $(p_{11},p_{12})(r^*,0) = 0$, i.e. there is a defect on $\theta = 0$. As $r$ increases from $0$ to $1$, $p_{11}(-r,0)-\frac{B}{C}c(r)$ is monotonically increasing from $0$ to $\overline{\alpha}_2 = \frac{B}{2C}(1-\overline{a}/2)\cos(\pi) = \frac{B}{2C}$, since $\overline{a} = 0$ at the middle of $D_2$. So with $\gamma^*=0$, under the SC mapping $f^{-1}$, the defect set of the limiting profile consists of a single isolated point, on $E^C_3$.

Using the rotation-based relations in \eqref{p11_rotation}, with $\gamma^* = 0$, $\theta =0$, $n = 1$, and $(p_{11},p_{12})(r^*,0) = 0$, we obtain $(p_{11},p_{12})(r^*e^{i2\pi/3}, \pi/3)=(p_{11},p_{12})(r^*e^{i2\pi/3},-2\pi/3) = 0$, i.e.,  for $\gamma^* = \pi/3$, there is a unique defect on $\theta = 2\pi/3$, and consequently for $\psi = -2\pi/3$ in the limiting profile $(P_{11},P_{12})(\rho e^{i\psi},\gamma^*)$.
For $\gamma^* \in \left(0, \pi/3 \right)$, the defect smoothly rotates between $\psi=0$ and $\psi = -2 \pi/3$ in the limiting profile (see Fig. \ref{fig:triangle}) and this is a good example of how $\gamma^*$ can tune the location of the defect for the limiting profile. 
In the first row of Fig.~\ref{fig:triangle}, we plot the solution of \eqref{eq:map_zero}, for $K=3$. The domain is a unit disc $D$, since we nondimensionalized the rLdG energy in \eqref{rLdG_energy}. 
In the second row of Fig.~\ref{fig:triangle}, we plot the corresponding solution of \eqref{zero_euler} on $E_3^C$, and again the polygonal hole is re-scaled to have unit length. The computational domain is the same for all values of $\lambda$, because of the rescaling in \eqref{p_energy}. In Fig.~\ref{fig:triangle} and subsequent figures, the color bar labels the order parameter, $s = \sqrt{p_{11}^2 + p_{12}^2}$, and the white lines correspond to the nematic director $\mathbf{n} = (\cos(arctan(p_{12}/p_{11})/2),\sin(arctan(p_{12}/p_{11})/2))$, where $p_{11}$ and $p_{12}$ are the two scalar fields in \eqref{eq:map_zero}. In Fig.~\ref{fig:triangle_error}, we plot the difference between the limiting solution of \eqref{zero_euler} and the solution of the Euler-Lagrange equations \eqref{Euler_Lagrange}, with $K=3$ and $\lambda^2 = 0.01$. The error is proportional to $\lambda^2$ and we deduce that the limiting solution is a good approximation to the solutions of \eqref{Euler_Lagrange} for small values of $\lambda^2$. This can also be formally justified using further asymptotic analysis. The remaining figures in this Section have been computed with the same numerical method, with the same details. The numerical details are given in Section 4.
\begin{figure}
\centering
        \includegraphics[width=0.5\columnwidth]{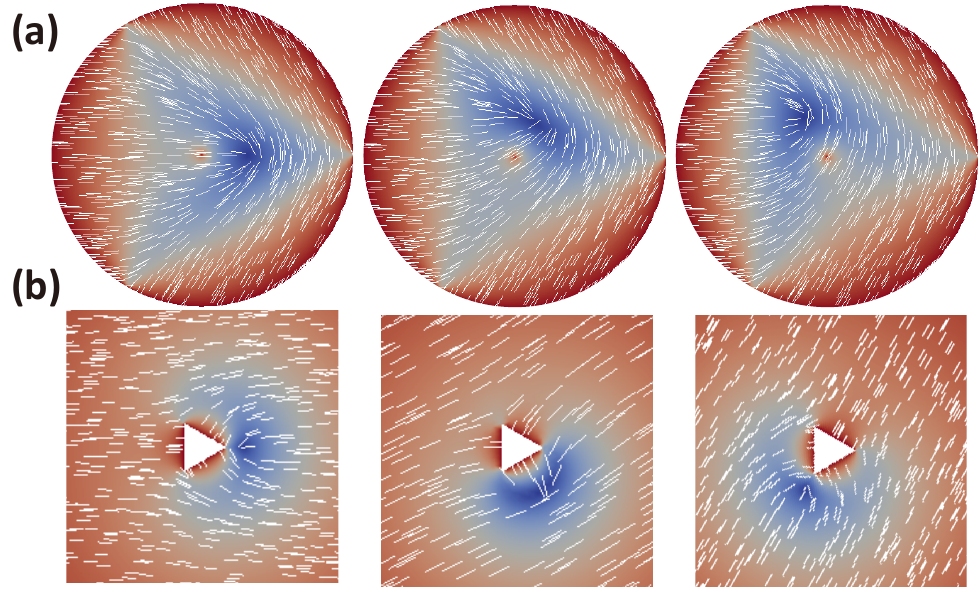}
        \caption{(a) The mapped solution of \eqref{eq:map_zero} on $D$ with $K = 3$ and (b) the numerical  limiting solution of \eqref{zero_euler} on $[-10,10]^2\backslash E_3$ as an approximation to $E_3^C$. From left to right, the boundary condition at infinity is $\gamma^*=0$, $\pi/6$, and $\pi/3$. In this figure and all subsequent figures where the value of $B$ and $C$ are required, we have $B = 0.64 \times 10^4 N/m^2$ and $C = 0.35 \times 10^4 N/m^2$. The colour bar is the scalar order parameter, $s = \sqrt{p_{11}^2 + p_{12}^2}$, and the white lines correspond to the nematic director $\mathbf{n} = (\cos(arctan(p_{12}/p_{11})/2),\sin(arctan(p_{12}/p_{11})/2))$.
}
        \label{fig:triangle}
\end{figure}
\begin{figure}
\centering
        \includegraphics[width=0.8\columnwidth]{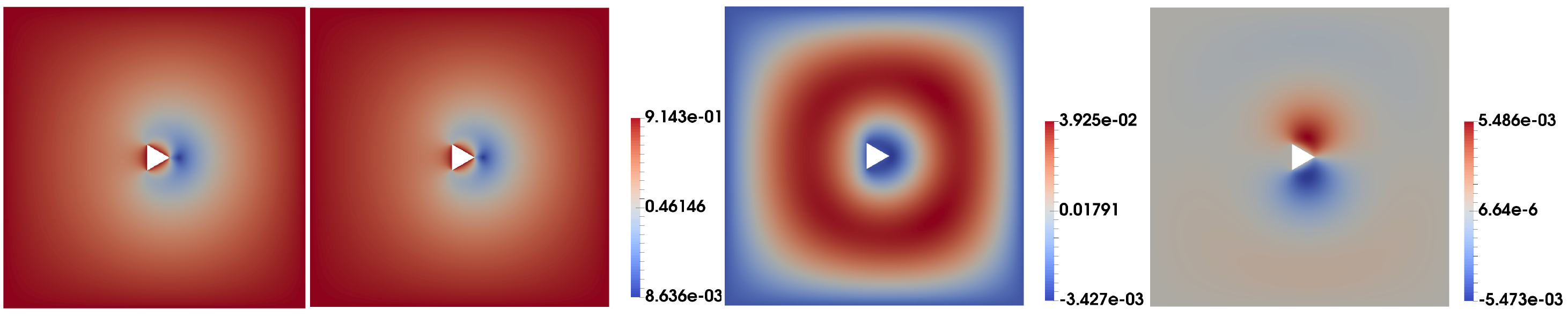}
        \caption{From left to right: the nematic order, $s$, of the limiting solution $(P_{11}^0,P_{12}^0)$ for \eqref{zero_euler} on $[-10,10]^2\backslash E_3$; the solution $(P_{11}^1,P_{12}^1)$ of Euler-Lagrange equations \eqref{Euler_Lagrange} with $\lambda^2 = 0.01$; comparison between the exact and limiting solutions: $|\mathbf{P}^1|^2/2-|\mathbf{P}^0|^2/2 = (s^1)^2-(s^0)^2$; the difference in the nematic director $P^1_{11}P^0_{12}  -  P^1_{12}P^0_{11} = s^0s^1\sin(2\gamma^0 - 2\gamma^1)$.
        } 
        \label{fig:triangle_error}
\end{figure}

\subsubsection{The limiting profile for $E^C_4$ - exterior of a unit square}
In \cite{kraljmajumdar2014}, \cite{canevari2017order}, the authors show that the unique limiting rLdG minimiser on $E_4$ with tangent boundary conditions, is the WORS solution in the $\lambda \to 0$ limit; the label ``WORS" has been proposed by the authors in recognition of the fact that the WORS supports two nodal or defect lines along the square diagonals and is an example of widely studied \emph{order reconstruction solutions in LdG theory. The interested reader is referred to these papers, \cite{kraljmajumdar2014}, \cite{canevari2017order}, for more details.} We study the analogous problem on $E_4^C$ with homeotropic boundary conditions, to study whether line defects survive on $E_4^C$ for different choices of $\gamma^*$. 

\begin{proposition}\label{square}
The solution of (\ref{zero_euler}) on $E^C_4$ has two line defects for $\gamma^* = \pi/4 + n\pi/2$, $n\in\mathbb{Z}$, and no interior defects otherwise.
\end{proposition}
\begin{proof}
Substituting $K=4$ into the Dirichlet boundary condition \eqref{pb}, $\overline{\alpha}_k\equiv 0$, $k = 1,\cdots, 4$.
According to \eqref{zero_solution_p11}, the solution of \eqref{eq:map_zero} on $E^C_4$ is given by
\begin{equation}
p_{11}(re^{i\theta},\gamma^*) = \frac{B}{2C}\cos(2\gamma^*)c(r).
\end{equation}
As in Proposition \ref{gamma_infinity_restriction}, we can assume that $\gamma\in[0,\pi/4]$.
The zero set of $p_{11}$ is $\emptyset$ for any $\gamma\in[0,\pi/4)$. For $\gamma^* = \pi/4$, we have $p_{11}(re^{i\theta},\pi/4)\equiv 0$.

Substituting $K = 4$, into \eqref{zero_solution_p12}, we have \begin{equation}\label{r_theta}
p_{12}(re^{i\theta},\gamma^*) - \frac{B}{2C}\sin(2\gamma^*)c(r)= -\frac{1}{2\pi}\sum_{k =1}^4\int_{\pi(k-1)/2}^{\pi k/2}\overline{\beta}_k\frac{1-r^2}{1-2r\cos(\phi-\theta)+r^2}d\phi.
\end{equation}
which only depends on $r$ and $\theta$.
In the following, we prove that with $\gamma^* = \pi/4$, $p_{12}(re^{i\theta},\pi/4)-\frac{B}{2C}c(r)$ is monotonically decreasing in the $r$-direction. Firstly, we prove that on $\theta = \pi$ i.e., $y = 0$ and $x\leq 0$, $p_{12}(re^{i\pi},\gamma^*)-\frac{B}{2C}\sin(2\gamma^*)c(r) \equiv 0$.
Using \eqref{p12_reflection}, we obtain
\begin{equation}\label{zero1}
p_{12}(re^{i \pi},\gamma^*)-\frac{B}{2C}\sin(2\gamma^*)c(r)= -(p_{12}(re^{-i \pi},-\gamma^*)-\frac{B}{2C}\sin(-2\gamma^*)c(r)).
\end{equation} 
Since $p_{12}(re^{i\theta},\gamma^*) - \frac{B}{2C}\sin(2\gamma^*)c(r)$ in \eqref{r_theta} only depends on $r$ and $\theta$, we have $p_{12}(re^{i\theta},\gamma^*) - \frac{B}{2C}\sin(2\gamma^*)c(r) = p_{12}(re^{i\theta},-\gamma^*) - \frac{B}{2C}\sin(-2\gamma^*)c(r),
$
and subsequently
\begin{equation}\label{zero2}
p_{12}(re^{-i \pi},-\gamma^*)-\frac{B}{2C}\sin(-2\gamma^*)c(r) = p_{12}(re^{-i \pi},\gamma^*)-\frac{B}{2C}\sin(2\gamma^*)c(r) = p_{12}(re^{i \pi},\gamma^*)-\frac{B}{2C}\sin(2\gamma^*)c(r).
\end{equation}
Combining \eqref{zero1} and \eqref{zero2}, we get the desired conclusion,
\begin{equation}\label{zero3}
p_{12}(re^{i \pi},\gamma^*)-\frac{B}{2C}\sin(2\gamma^*)c(r)=0.
\end{equation}

Similarly, we prove that on $\theta = 3\pi/2$, i.e., $x=0$ and $y\leq 0$, $p_{12}(re^{i3\pi/2},\gamma^*)-\frac{B}{2C}\sin(2\gamma^*)c(r) \equiv 0$.
Using \eqref{p12_rotation} with $n = 1$, $K = 4$, $\theta = \pi$, and $\gamma = \pi/4$, we have
\begin{equation} \nonumber
p_{12}(re^{i 3\pi/2},-\pi/4) = \cos(\pi)p_{12}(re^{i \pi},\pi/4) = -p_{12}(re^{i \pi},\pi/4).
\end{equation}
From \eqref{zero3}, $p_{12}(re^{i \pi},\pi/4) = \frac{B}{2C}\sin(\pi/2)c(r)$ and hence, 
\begin{align}
p_{12}(re^{i3\pi/2},\gamma^*)-\frac{B}{2C}\sin(2\gamma^*)c(r)&=p_{12}(re^{i3\pi/2},-\pi/4)-\frac{B}{2C}\sin(-\pi/2)c(r)\nonumber \\
&=  -p_{12}(re^{i \pi},\pi/4) + \frac{B}{2C}\sin(\pi/2)c(r) \equiv 0.
\end{align}
\begin{figure}
\centering
        \includegraphics[width=0.3\columnwidth]{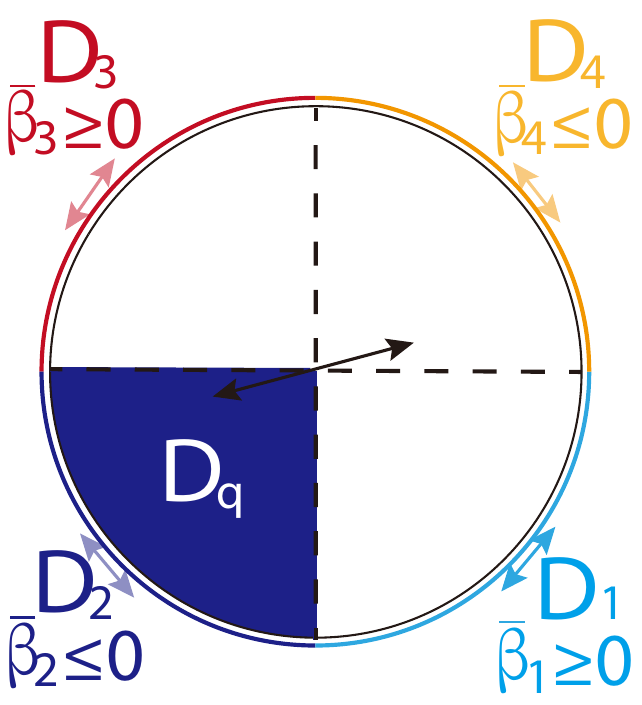}
        \caption{The mapped domain $D$, with $K = 4$. We show the sign of $\overline{\beta}_k$, $k = 1,\cdots, 4$, the boundary condition of $p_{12}$ on $D_k$, the segment of $\partial D$, and $\overline{\beta}_2, \overline{\beta}_4\leq 0$, $\overline{\beta}_1, \overline{\beta}_3\geq 0$. The quarter of disc $D_q = D\cap\{x\leq 0\}\cap\{y\leq 0\}$ is shown in dark blue.}
        \label{fig:disc_4}
\end{figure}

On the quadrant $\theta\in[\pi,3\pi/2]$ or the quarter disc,  $D_q = D\cap\{x\leq 0\}\cap\{y\leq 0\}$ (see Fig. \ref{fig:disc_4}), the function, $p_{12}(r e^{i\theta},\pi/4)-\frac{B}{2C}\sin(\pi/2)c(r)$, is the solution of 
\begin{align}
&\Delta u_a = 0,\ on\ D_q,\nonumber\\
&u_a = \overline{\beta}_2,\ on\ D_2,\nonumber\\
&u_a = 0,\ on\ x = 0\ and\ y = 0.\label{laplace}
\end{align}
where $D_2$ is the segment of $\partial D$ with $x\leq0$ and $y\leq 0$, $\overline{\beta}_2\leq 0$ is the corresponding Dirichlet boundary condition of $p_{12}$ on $D_2$.
Following the arguments from Lemma 4.5 in \cite{canevari2017order}, the first step is to show that
$u_a$ is non-increasing in the $r$ direction, i.e.
\begin{gather}\label{non-decreasing}
u_a(r,\theta)\geq u_a(\tau r, \theta),\ \forall \tau>1,\ (r,\theta)\in D_q.
\end{gather}
We consider Problem \eqref{laplace_tau}, the analogue of Problem \eqref{laplace} on the extended domain, $D_q^{\tau} = (\tau r,\theta)$ where $(r,\theta)\in D_q$,
\begin{align}
&\Delta u_a^{\tau} = 0,\ on\ D_q^{\tau},\nonumber\\
&u_a^{\tau} = \overline{\beta}_2,\ on\ D_2^{\tau},\nonumber\\
&u_a^{\tau} = 0,\ on\ x = 0\ and\ y = 0.\label{laplace_tau}
\end{align}
where $D_2^{\tau} = (\tau r,\theta)$, with $(r,\theta)\in D_2$.
By scaling invariance, the unique non-negative solution to Problem \eqref{laplace_tau} is given by
\begin{gather}\label{translation}
u_a^{\tau}(\tau r,\theta):=u_a(r,\theta)\ for\ any\ (r,\theta)\in D_q.
\end{gather}
Moreover, the function
\begin{align}
u_a^s(r,\theta) = \begin{cases} 
&u_a\ on\ D_q,\\
&\overline{\beta}_2(\theta),\ on\ \{r>1\}\cup\{x\leq 0\}\cup\{y\leq 0\}\\
\end{cases}
\end{align}
is a subsolution of \eqref{laplace_tau}, with $\frac{\partial^2 \overline{a}}{\partial \theta^2}\geq 0$, so that 
\begin{gather}
u_a^{\tau}(\tau r,\theta)\geq u_a^s(\tau r,\theta)\ for\ any\ (r,\theta)\in D_q\ s.t.\ (\tau r,\theta)\in D_q
\end{gather}
Recalling \eqref{translation} and using $u_a^s = u_a$ on $D_q$, we conclude the proof of \eqref{non-decreasing}.
Let $v_a = \partial u_a/\partial r$. By \eqref{non-decreasing}, $v_a\leq 0$ on $D_q$; we want to prove that the strict inequality holds. We differentiate Equation \eqref{laplace} with respect to $r$, so that
$\nabla v_a = 0\ on\ D_q.$
By the strong maximum principle, we deduce that either $v_a\equiv 0$ or $v_a<0$ in $D_q$. The first possibility is clearly inconsistent with the boundary conditions for $u_a$ in (\ref{laplace}), and hence, $v_a$ must be strictly negative inside $D_q$.

Therefore, on the quadrant $\theta\in[\pi,3\pi/2]$, $p_{12}(re^{i\theta},\pi/4)-\frac{B}{2C}\sin(\pi/2)c(r)$ is monotonically decreasing in $r$ directions from $0$ to $\overline{\beta}_2$. 
Using analogous arguments,  $p_{12}(re^{i\theta},\pi/4)-\frac{B}{2C}\sin(\pi/2)c(r)$ is monotonically decreasing in the $r$-direction, from $0$ to $\overline{\beta}_4$ for on $\theta\in[0,\pi/2]$, and  monotonically increasing in $r$, from $0$ to $\overline{\beta}_1$ or $\overline{\beta}_3$, on the quadrants $\theta\in[\pi/2,\pi]\cup [3\pi/2,2\pi]$.

Moreover, $\frac{B}{2C}\sin(\pi/2)c(r)$ is monotonically decreasing from $\frac{B}{2C}$ to $0$, hence there exists $r^*(\theta)$ for each $\theta\in[0,\pi/2]\cup[\pi,3\pi/2]$ s.t. $p_{12}(r^*(\theta)e^{i\theta},\pi/4) = 0$, so that there are two line defects in limiting profile for $E^C_K$ with $\gamma^* = \pi/4$. Using \eqref{p11_rotation} and \eqref{p12_rotation}, $(p_{11},p_{12})(re^{i\theta + i 2\pi n/K},\gamma^*-\frac{2\pi n}{K})$, $n\in\mathbb{Z}$, can be obtained by rotating $(p_{11},p_{12})(re^{i\theta},\gamma^*)$. Hence, there are two line defects for the limiting profile $(P_{11},P_{12})(\rho e^{i\psi},\gamma^*)$ on $E^C_4$, for $\gamma^* = \pi/4 + n\pi/2$, $n\in\mathbb{Z}$.
\end{proof}
We plot the mapped solution of \eqref{eq:map_zero} on $D$ with $K = 4$ and corresponding limiting solution of \eqref{zero_euler} on $E_4^C$ with $\gamma^* = 0$ and $\pi/4$ in Fig. \ref{fig:square} respectively.
On $E^C_4$, with $\gamma^* = \pi/4$, there are two line defects; for $\gamma^* = 0$, there are no interior defects on $E^C_4$ although there are defects at the square vertices where the nematic director cannot be uniquely defined. 
\begin{figure}
\centering
        \includegraphics[width=0.6\columnwidth]{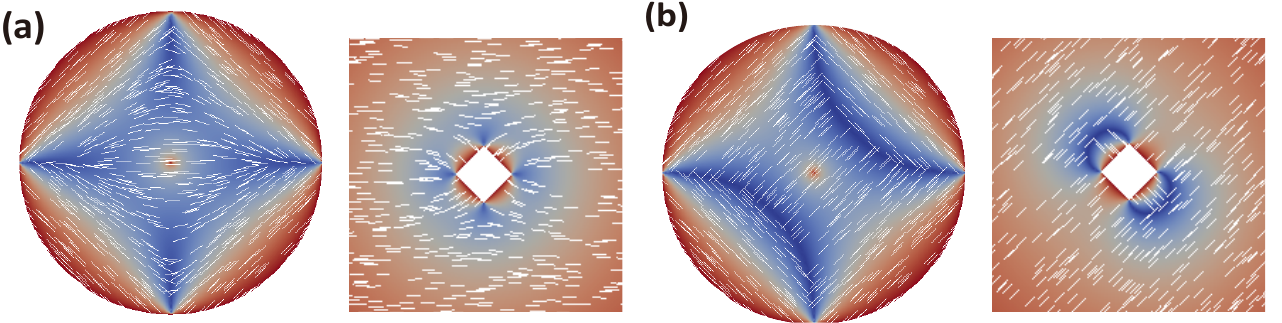}
        \caption{The mapped solution of \eqref{eq:map_zero} on $D$ with $K = 4$ and the corresponding limiting solution for \eqref{zero_euler} on $E_4^C$} with (a) $\gamma^* = 0$ and (b) $\gamma^* = \pi/4$.
        \label{fig:square}
\end{figure}
\subsubsection{The limiting profile for $E^C_K$ with $K>4$}

\begin{lemma}\label{lemma1}
For $\gamma^*\in[0,\pi/K]$, $K\in\mathbb{Z}$ and $K>4$, the solution of (\ref{eq:map_zero}) has no defect on $D_{q1} =  D\cup \{2\pi-\frac{\pi}{K}\leq\theta\leq 2\pi\}$ and $D_{q2} =  D\cup \{\pi-\frac{\pi}{K}\leq\theta\leq\pi \}$. 
\end{lemma}
\begin{proof}
As in Proposition \ref{square}, we can prove that on $D_{q1}$ and $D_{q2}$, $p_{12}(re^{i\theta},\gamma^*)-\frac{B}{2C}\sin(2\gamma^*)c(r)$ increases monotonically in the $r$-direction from $0$ to $\beta_1>0$, as $r$ increases from $0$ to $1$.
Hence, $p_{12}(re^{i\theta},\gamma^*)-\frac{B}{2C}\sin(2\gamma^*)c(r)$ is positive for $r>0$ on $D_{q1}$ and $D_{q2}$.
For $K>4$ and $K\in\mathbb{Z}$, $\frac{B}{2C}\sin(2\gamma^*)c(r)$ is non-negative, and is zero if and only if $r=1$, i.e., on $\partial D$. 
Hence, $p_{12}$ is nonzero on $D_{q1}$ or $D_{q2}$, i.e., there is no defect on $D_{q1}$ or $D_{q2}$.

\end{proof}
As $p_{12}$ and $p_{11}$ are analytic on $D$, the angle $\gamma = arctan(p_{12}/p_{11})/2$ is continuous.  
We can separate the disc into four parts by two smooth curves (see Fig. \ref{fig:domain_partition}): on $curve_1$ (from $z=e^{i(K-1)\pi/K}$ to $z= 0$ to $z=1$), $\gamma$ decreases from $\pi/K$ to $\gamma^*$ to $0$ monotonically, and on $curve_2$ (from $z=-1$ to $z= 0$ to $z=e^{-i\pi/K}$), $\gamma$ increases from $0$ to $\gamma^*$, and then to $\pi/K$  monotonically. Since there is no defect on $D_{q1}$ and $D_{q2}$ as proved in Lemma \ref{lemma1}, there is no defect on $curve_1$ and $curve_2$. 
We define $D_u$ to be the part of $D$ above $curve_1$, and $D_b$ to be the part of $D$ below $curve_2$. The boundary $\partial D_u$ ($\partial D_b$) consists of a segment of $\partial D$ and the interior curve $curve_1$ ($curve_2$). In the following, we prove that the director angle $\gamma$ is strictly monotonic on $\partial D$ and subsequently, that $\gamma$ is strictly monotonic on either $\partial D_u$ or $\partial D_b$, and rotates by $\pi$. 
\begin{figure}
\centering
        \includegraphics[width=0.3\columnwidth]{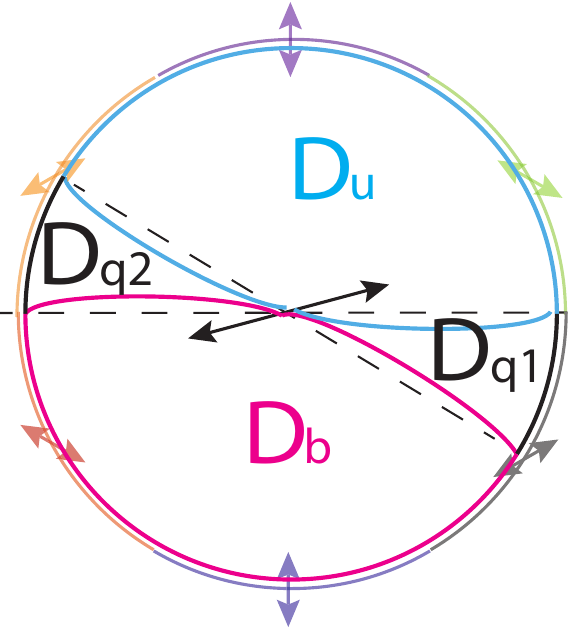}
        \caption{The partition of domain $D$. $D_u$ and $D_b$ are surrounded by blue and pink curves, respectively. $D_{q1} =  D\cap \{2\pi-\frac{\pi}{K}\leq\theta\leq2\pi\}$ and $D_{q2} =  D\cap \{\pi-\frac{\pi}{K}\leq\theta\leq\pi\}$ are domains between the two dashed black lines and solid black curves in $D$. The black two-head arrow in the center of $D$ represents the direction $\mathbf{n}$ of constraint $(p_{11},p_{12})(0,\gamma^*) = (\cos2\gamma^*,\sin2\gamma^*)$. The arcs with different colors in the outer circle represents the segment of boundary $D_k$, $k = 1,\cdots,K$. The two-head arrows on the boundary $D_k$ represent the direction $\mathbf{n}$ of boundary conditions; here, we take hexagon as an example.}
        \label{fig:domain_partition}
\end{figure}
\begin{lemma}\label{lemma2}
On $\partial D$, the angle $\gamma = arctan(p_{12}/p_{11})/2$ is strictly monotonic, and there is no defect.
\end{lemma}
\begin{proof}
On $D_{1}\cap\theta\in[2\pi-\pi/K,2\pi]$, the boundary condition is
\begin{align}
&(p_{11b},p_{12b}) = (\overline{\alpha}_1,\overline{\beta}_1)\nonumber\\
& = \frac{B}{2C}\left((1-\frac{\overline{a}}{2})\cos\left(\frac{2\pi}{K}\right) + \frac{\overline{a}}{2}\cos\left(\frac{-2\pi}{K}\right)\right.,\left.(1-\frac{\overline{a}}{2})\sin\left(\frac{2\pi}{K}\right) + \frac{\overline{a}}{2}\sin\left(\frac{-2\pi}{K}\right)\right)\nonumber\\
&= \frac{B}{2C}\left(\cos\left(\frac{2\pi}{K}\right),(1-\overline{a})\sin\left(\frac{2\pi}{K}\right)\right).\nonumber
\end{align} 
The angle, $\gamma = arctan((1-\overline{a})tan(\frac{2\pi}{K}))/2$ monotonically increases from $0$ to $\frac{\pi}{K}$ as $\overline{a}$ decreases from $1$ to $0$, i.e., $\theta$ decreases from $2\pi$ to $2\pi-\frac{\pi}{K}$. The angle $\gamma$ on $\partial D$ rotates by $-2\pi k/K$, as $\theta$ rotates by $2\pi k/K$ and and reflects about axis $\theta = n\pi/K$, $n\in\mathbb{Z}$. Hence, $\gamma$ on $D_b\cap\partial D$ increases monotonically from $\pi/K$ to $\pi$, in a clockwise fashion. Since $\gamma$ on $curve_2$ increases monotonically from $0$ to $\pi/K$, then $\gamma$ on $\partial D_b$ rotates by $\pi$ (analogous remarks apply to $D_u$).

Consider $p_{11b}^2 + p_{12b}^2$ on $D_1$, $\theta\in[2\pi-\pi/K,2\pi]$, for $K>4$ i.e.
\begin{align}
&p_{11b}^2 + p_{12b}^2 = \frac{B^2}{4C^2}\left((1-\frac{\overline{a}}{2})\cos\left(\frac{2\pi}{K}\right) + \frac{\overline{a}}{2}\cos\left(\frac{-2\pi}{K}\right) \right)^2 + \frac{B^2}{4C^2}\left((1-\frac{\overline{a}}{2})\sin\left(\frac{2\pi}{K}\right) + \frac{\overline{a}}{2}\sin\left(\frac{-2\pi}{K}\right)\right)^2\nonumber \\
& = \frac{B^2}{4C^2}\left((\overline{a}^2-2\overline{a})\sin^2\frac{2\pi}{K} + 1 \right)>0.\nonumber
\end{align} 
From the rotation-based relations in \eqref{alpha_rotation} and \eqref{beta_rotation},  $p_{11b}^2 + p_{12b}^2$ is nonzero on $D_k$, $\theta\in[2\pi(K-k+1)/K - \pi/K,2\pi(K-k+1)/K]$, $k = 1,\cdots,K$. From the reflection relations in \eqref{alpha_reflection} and \eqref{beta_reflection}, $p_{11b}^2 + p_{12b}^2$ is nonzero on $D_k$, $k = 1,\cdots,K$. 
Hence, there are no defects (or zeros of $\mathbf{p}$) on $\partial D$, for any regular polygon with $K>4$ edges.
\end{proof}


\begin{theorem}\label{general}
If $\mathbf{p}$ is a minimizer of $J = \int_{D_u}|\nabla\mathbf{p}|^2 \textrm{dA}$ in $\mathcal{A} = \{\mathbf{p}\in W^{1,2}(D_u;\mathbb{R}^2):\mathbf{p} = \mathbf{p}_{bu}\ on\ \partial D_u\}$, then there is a unique point defect in $D_u$. Similarly, there is a unique point defect in $D_b$. The solution of \eqref{eq:map_zero} has two zeroes or point defects and hence, under the SC mapping $f$ \eqref{eq:SC}, the limiting profile has two point defects in exterior of regular polygon with $K>4$ edges.
\end{theorem}
\begin{proof}
Let $Y(l)$ for $0\leq l \leq L$ be a one-to-one parametrization of $\partial\D_u$ with respect to arclength. Considering the Dirichlet data, $(p_{11bu},p_{12bu})$, in $C^{2,\mu}(\partial \D_u;\mathbb{R}^2)$, we have
\begin{equation}
(p_{11bu},p_{12bu})(Y(l)) = (s(l) cos2\gamma(l), s(l) sin2\gamma(l)).
\end{equation}
As proved in Lemma \ref{lemma1} and \ref{lemma2}, we have
\begin{gather}
s(l) > 0,\ \gamma'(l)\neq 0\ for\ l\in[0,L],\  and\ |2\gamma(L)- 2\gamma(0)| = 2\pi.
\end{gather}
For $\alpha\in\mathbb{R}$ and $\mathbf{p}$, a minimizer for $J$ in $\mathcal{A}$, set
\begin{equation}
w_{\alpha}(X) =  -p_{11}(X)\sin(\alpha) + p_{12}(X)\cos(\alpha).
\end{equation}
Define the nodal set of $w_{\alpha}$ by
\begin{equation}
N_{\alpha}\equiv\{X\in\overline{D}_u:w_{\alpha}(X) = 0\}.
\end{equation}
Note that for any pair $\alpha_1$, $\alpha_2$ with $0 \leq \alpha_1 < \alpha_2 < \pi$, the set of zeros of $\mathbf{p}$ is just
$\Gamma(\mathbf{p}) = N_{\alpha_1}\cap N_{\alpha_2}$. Thus as $\alpha$ varies, $\Gamma(\mathbf{p})$ is the subset of $N_{\alpha}$ that remains fixed.
Analogous to Lemma 2.2 in \cite{baumanowensphillips}, we can prove for each $\alpha\in[0,2\pi]$, $N_{\alpha}$ is a $C^1$ imbedded curve in $\overline{D}_u$, which enters and exits $\overline{D}_u$, at distinct points of $\partial D_u$.
The following proof is identical to Theorem 2.3 in \cite{baumanowensphillips}. Let $\{P_1, P_2\} = N_{0}\cap \partial\D_u$ and assume without loss of generality that $\sin(2\gamma)\vert_{P_1} = 0$ with $\cos(2\gamma)\vert_{P_1}<0$ and $\sin(2\gamma)\vert_{P_2} = 0$ with $\cos(2\gamma)\vert_{P_2}>0$. The points $P_1$ and $P_2$ partition $\partial D_u$ into two arcs, $(\partial D_u)^-\equiv\{X\in\partial D:\sin(2\gamma)\vert_X\leq 0\}$ and $(\partial D_u)^+\equiv\{X\in\partial D:\sin(2\gamma)\vert_X\geq 0\}$. Starting at $P_1$ and moving along $N_0$, let $X_0$ be the first point reached in $\Gamma(\mathbf{p})$. Since $X_0\in N_{\alpha}$ for all $\alpha$, each $N_{\alpha}$ can be parametrized by arclength, $X = X(\tau, \alpha)$, such that
\begin{gather}
a(\alpha) < \tau < b(\alpha),\ a(\alpha) < 0, X(a(\alpha),\alpha)\in(\partial\D_u)^-,\ b(\alpha)>0,\ X(b(\alpha),\alpha)\in(\partial\D_u)^+,\nonumber\\
\left\vert\frac{\partial X}{\partial \tau}\right\vert = 1,\ and\ X(0,\alpha) = X_0.\nonumber
\end{gather}
Set $\mathcal{D} = \{(\tau,\alpha):a(\alpha)\leq\tau\leq b(\alpha),0\leq\alpha\leq\pi\}$. In part 1 of Theorem 2.3, the authors proved $X\in C^1(\mathcal{D})$, and $a(\alpha)$ and $b(\alpha)$ are in $C^1([0,\pi])$. In part 2, they have that for all $\alpha$ in $[0,\pi]$ and $\tau<0$, $X(\tau,\alpha)\notin\Gamma(\mathbf{p})$. In particular, $X(\tau,\pi)\notin\Gamma(\mathbf{p})$. In part 3, now $X(\tau, 0)$ and $X(\tau, \pi)$ are each parametrizations of $N_0$ in opposite directions. Since $X_0$ is the first point in each direction along $N_0$ that is in $\Gamma(\mathbf{p})$, we conclude that $\Gamma(\mathbf{p}) = {X_0}$, i.e., there is a unique point defect in $D_u$. Similarly, there is a unique point defect in $D_b$. Under the SC mapping $f$ \eqref{eq:SC}, there are two point defects in exterior of regular polygon with $K>4$ edges.
\end{proof}
\begin{rem}\label{remark2}
Theorem \ref{general} does not work for $E^C_K$ with $K = 3$ or $4$.
For $K = 3$, since there exists a $r^*$ s.t. $(p_{11},p_{12})(r^*e^{i2\pi/3}, \pi/3) = (p_{11},p_{12})(r^*, 0) = 0$, we may have a defect on $D_{q1}$ or $D_{q2}$ which does not satisfy Lemma \ref{lemma1}. 
For $K = 4$, with the given boundary conditions, we have defects at the four vertices, $(p_{11},p_{12}) = \frac{B}{2C}\left(cos(\frac{2\pi}{K}),(1-\overline{a})\sin(\frac{2\pi}{K})\right) = (0,1-\overline{a}(0)) = (0,0)$, which does not satisfy Lemma \ref{lemma1} and \ref{lemma2}.
\end{rem}
\begin{figure}
\centering
        \includegraphics[width=0.8\columnwidth]{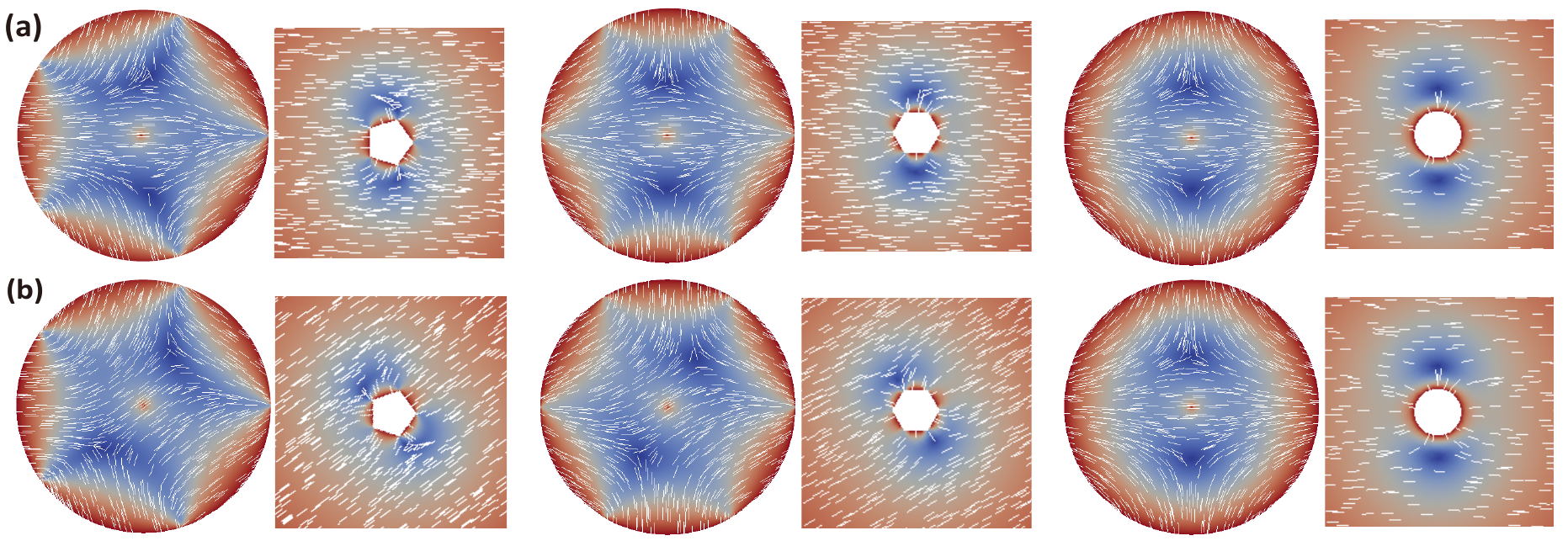}
        \caption{The mapped solution of \eqref{eq:map_zero} on $D$ with $K = 5,6,\infty$ and the corresponding limiting solution of \eqref{zero_euler} on $E_K^C$, the complement or exterior domain of a pentagon, hexagon and disc with (a) $\gamma^* = 0$ and (b) $\gamma^* = \pi/K$.} 
       \label{fig:general}
\end{figure}
The mapped solution of \eqref{eq:map_zero} with $K = 3,4,5,6,\infty$, and the corresponding solutions of \eqref{zero_euler} in the exterior of a triangle, square, pentagon, hexagon and disc, with $\gamma^* = 0$ and $\gamma^*=\pi/K$, are plotted in Fig. \ref{fig:triangle}, Fig. \ref{fig:square} and Fig. \ref{fig:general} respectively. With even $K$, the solutions have the symmetry property in Corollary \ref{remark1}.
With $\gamma^* = 0$ and $\pi/K$, the solutions have the symmetry properties in Corollary \ref{0_symmetry} and Corollary \ref{piK_symmetry}, respectively. The mapped solution has two $-1/2$ defects on two sides of $\theta = \gamma^*+\pi/2$ as in Theorem~\ref{general}, for $K>4$.

\section{The $\lambda\to\infty$ limit}\label{sec:infinity}
In this section, we use heuristics and numerical methods to study stable rLdG equilibria in the $\lambda \to \infty$ limit. This limit is relevant for polygonal holes whose edge length, $\bar{\lambda}$, is much greater than the nematic correlation length i.e. micron-scale holes. Consider the rLdG energy in \eqref{p_energy}; as $\lambda \to \infty$, minimisers of the rLdG energy will belong to the set of minimisers of the bulk energy, away from defects. Recall that the bulk energy density is given by
$$ f_b(\Pvec) = -\frac{B^2}{4C} \textrm{tr}\Pvec^2 + \frac{C}{4}\left(\textrm{tr}\Pvec^2 \right)^2.$$

More precisely, in Theorem 2 of \cite{Bronsard2016minimizers}, the authors prove that global minimisers of the rLdG energy in \eqref{p_energy}, converge strongly in $W^{1,2}$, to
\begin{equation}\label{P_infty}
        \mathbf{P}^{\infty} = \frac{B}{C}\left(\nvec^{\infty}\otimes \nvec^{\infty}-\frac{1}{2}\mathbf{I}_2\right),\nonumber
\end{equation}
where
\begin{equation}
\nvec^{\infty} = \left(\cos\gamma^*,\sin\gamma^*\right),
\end{equation}
and $\gamma^*$ is the solution of the following equations
\begin{equation}
        \begin{aligned}
            &\Delta\gamma = 0,\ on\ E_K^C\\
            &\gamma = \gamma_b,\ on\ \partial E_K,\\
            &\gamma = \gamma^*,\ |x|\to\infty,
        \end{aligned}
        \label{infty_euler}
\end{equation} everywhere away from the vertices of $E_K$. The convergence is actually stronger but that is not the focus of this paper.

\begin{figure}
    \begin{center}
        \includegraphics[width=0.3\columnwidth]{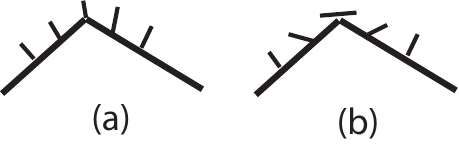}
        \caption{vertex classification according to director profile: (a) splay and (b) bend}
        \label{fig:vertices}
    \end{center}
\end{figure}

We use this straightforward characterisation to compute an (under)estimate of the number of stable rLdG equilibria on $E_C^K$, in the $\lambda \to \infty$ limit. Our discussion does not cover local energy minimisers but does serve to illustrate the rich prospects of multistability tailored by two model parameters - $K$ and $\gamma^*$. Multistability arises from the fact that for a given $\gamma^*$, there are multiple choices of $\gamma_b$ consistent with the homeotropic boundary conditions in \eqref{Pb_constant} as $\sigma\to 0$, and each choice of $\gamma_b$ yields a limiting rLdG minimiser profile in the $\lambda \to \infty$ limit. We estimate the number of limiting profiles based on two assumptions - (i)  we restrict $\gamma_b$ so that $\gamma$ rotates by either $2\pi/K$ or $2\pi/K-\pi$ at a vertex (see Fig. \ref{fig:vertices}; referred to as ``splay" and ``bend" vertices respectively). This rotation is required to match the homeotropic boundary conditions at the two intersecting edges. (ii) We assume that $\textrm{deg}\left(\mathbf{n}_b,\partial E_K^C\right) = 0$, where $\mathbf{n}_b = (\cos\gamma_b,\sin\gamma_b)$, so that if $n_b$ denotes the number of ``bend'' vertices, then we necessarily have 
\begin{equation}
   \textrm{deg}\left(\mathbf{n}_b,\partial E_K^C\right) = (K-n_b)\left(2\pi/K\right)+n_b\left(2\pi/K-\pi\right) = 0,\nonumber
\end{equation}
 equivalent to $n_b = 2$. The two assumptions ensure that $\mathbf{n}_b$ has the minimal rotation around $\partial E^C_K$ consistent with the homeotropic boundary conditions. We can of course, have scenarios wherein $\gamma_b$ rotates by a greater amount or when the degree of $\mathbf{n}_b$ is non-zero, where $\mathbf{n}_b$ is as above, but it is reasonable to conjecture that rLdG energy minimisers have the simplest admissible topology. Consider $E_K^C$; based on the assumptions above, we have $2$ bend vertices on $E_K^C$ for each $K \geq 3$ and hence, at least $\binom{K}{2}$ limiting profiles and $[K/2]$ classes of solutions which are not related by rotation or reflection in this limit. This is identical to the combinatorial estimate for the number of stable equilibria for the interior problem on $E_K$ in \cite{han2020pol}, in the $\lambda \to \infty$ limit. However, this estimate does not account for the constraint at infinity \eqref{constraint_infty}, and multistability is enhanced for the exterior problem compared to the interior problem, by virtue of $\gamma^*$.


\begin{figure} 
\centering
        \includegraphics[width=0.8\columnwidth]{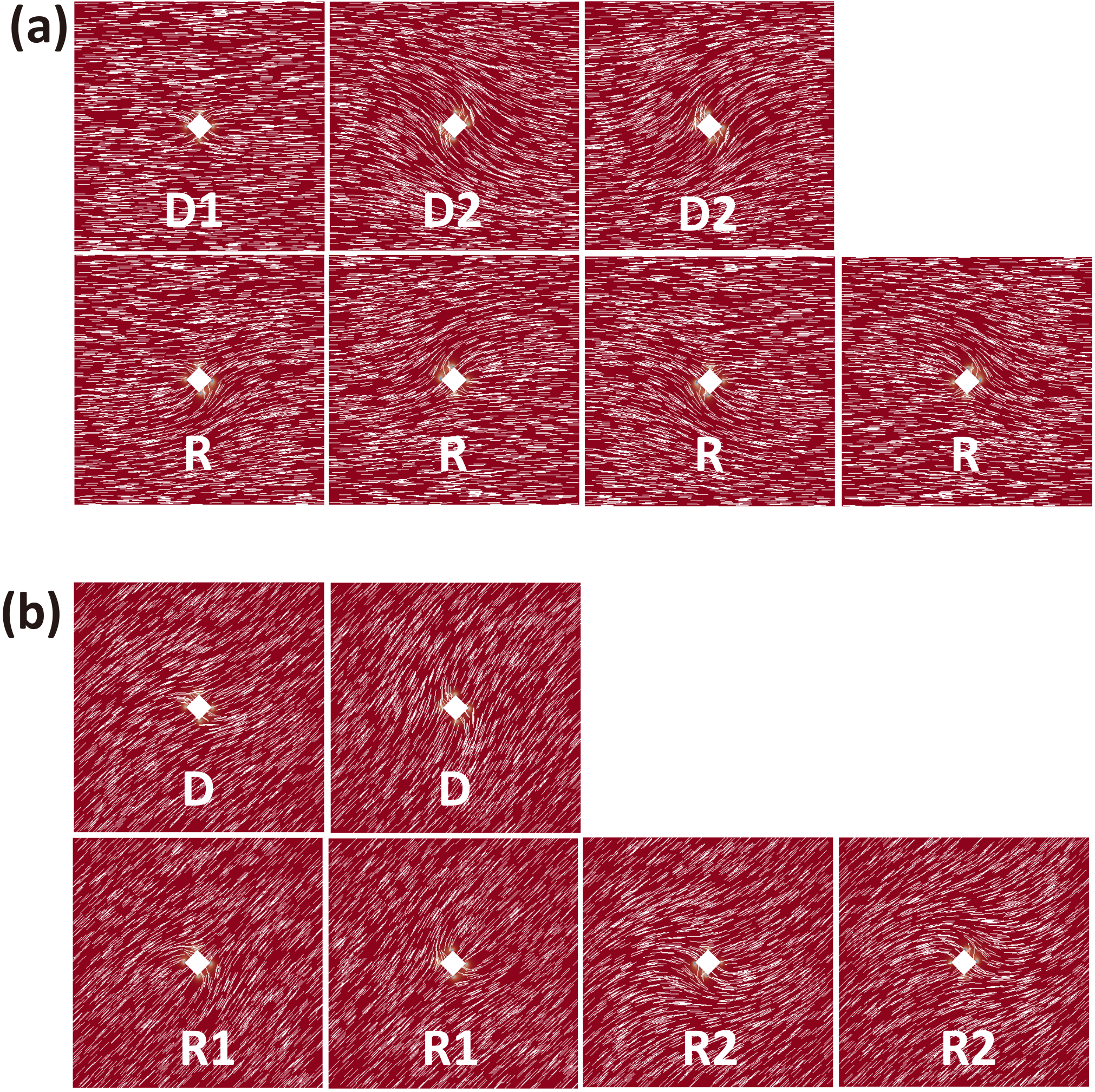}
        \caption{The limiting profiles for $E^C_4$, in the $\lambda\to \infty$ limit, with $\gamma^*= n\pi$ (a) and $\gamma^* = n\pi + \pi/4$ (b). The diagonal states are plotted in the first row of (a) and (b). The rotated states are plotted in the second row of (a) and (b). The white lines represent the nematic director $\mathbf{n} = (\cos(\gamma),\sin(\gamma))$, and the red color indicate the nematic order $s \equiv \frac{B}{2C}$.}
        \label{fig:dipole}
\end{figure}
For $K=4$, there are $6$ different choices for the two ``bend" vertices, (i) $2$ of which correspond to the two pairs of diagonally opposite vertices, (ii) $4$ of which correspond to pairs of adjacent vertices.
We refer to (i) as diagonal, \emph{D} states, (ii) as rotated \emph{R} states, following the nomenclature in \cite{luo2012}, \cite{tsakonas2007multistable}. 

The diagonal states are defined by the following choices of $\gamma_b$ on $E_4$:
\begin{gather}
\gamma_{b1} = \begin{cases}
&\pi/4,\ on\ C_1,\\
& -\pi/4,\ on\ C_2,\\
& \pi/4,\ on\ C_3,\\
& -\pi/4\ on\ C_4,
\end{cases}
\gamma_{b2} = \begin{cases}
&\pi/4,\ on\ C_1,\\
& 3\pi/4,\ on\ C_2,\\
& \pi/4,\ on\ C_3,\\
& 3\pi/4\ on\ C_4.
\end{cases}
\end{gather}

There are four different choices of $\gamma_b$, corresponding to the four \emph{R} rotated states, as enumerated below: 
\begin{gather}
\gamma_b = \begin{cases}
&\pi/4,\ on\ C_1,\\
& -\pi/4,\ on\ C_2,\\
& \pi/4,\ on\ C_3,\\
& 3\pi/4\ on\ C_4,
\end{cases}
\gamma_b = \begin{cases}
&\pi/4,\ on\ C_1,\\
& 3\pi/4,\ on\ C_2,\\
& \pi/4,\ on\ C_3,\\
& -\pi/4\ on\ C_4,
\end{cases}
\gamma_b = \begin{cases}
&\pi/4,\ on\ C_1,\\
& -\pi/4,\ on\ C_2,\\
& -3\pi/4,\ on\ C_3,\\
& -\pi/4\ on\ C_4,
\end{cases}
\gamma_b = \begin{cases}
&\pi/4,\ on\ C_1,\\
& 3\pi/4,\ on\ C_2,\\
& 5\pi/4,\ on\ C_3,\\
& 3\pi/4\ on\ C_4.
\end{cases}
\end{gather}  
For the interior problem on $E_4$, discarding rotation and reflection symmetries, there is simply one \emph{D} state and one \emph{R} state. This is not the case for the exterior problem on $E_4^C$, due to the choice of the $\gamma^*$. In Figure~\ref{fig:dipole}, we plot scalar order parameter and planar director of the reduced $\Pvec$-tensors defined in \eqref{P_infty} for the different choices of $\gamma_b$ defined above, for two different choices of $\gamma^*$. For $\gamma^* = n\pi$,  there are three distinct classes of solutions: $D1$, $D2$ and $R$ with $\gamma^* = n\pi$, and for $\gamma^* = \pi/4 + n\pi$, there are four distinct classes of solutions: $St$, $D$, $R1$, and $R2$. The label $St$ originates from the work in \cite{phillips2011texture}, which is the abbreviation of string defect mode.
\begin{figure*}
        \centering
        \includegraphics[width=\textwidth]{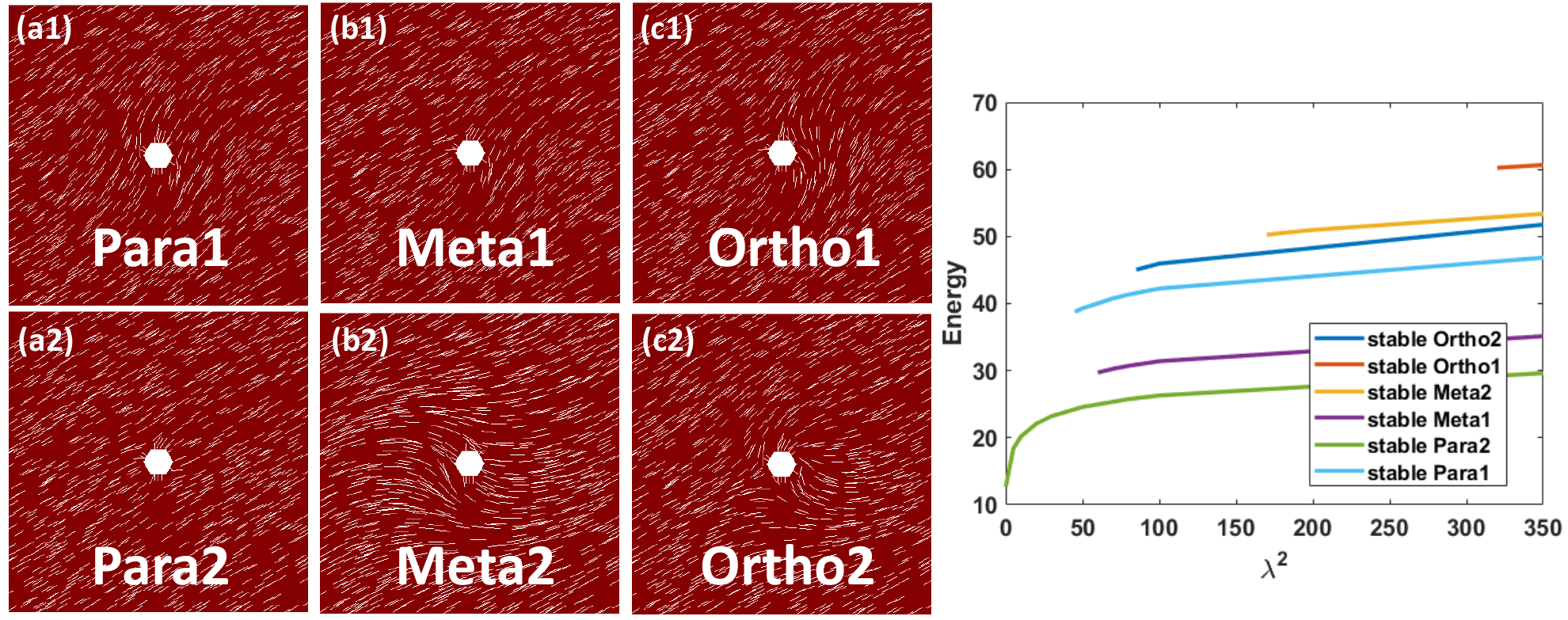}
        \caption{Left: distinct classes of solution of \eqref{infty_euler} which cannot be recovered from each other by reflection or rotation on the exterior of hexagon domain with $\gamma^* = \pi/6+n\pi$, (a1) $Para1$, (a2) $Para2$, (b1) $Meta1$, (b2) $Meta2$, (c1) $Ortho1$, (c2) $Ortho2$. Right: the energy plot of solutions in \eqref{Euler_Lagrange} versus $\lambda^2$.}
       \label{fig:dipole_hexagon}
\end{figure*}

We can repeat the same combinatorial arguments for $K=6$, to get simple estimates for the number of stable states on $E^C_6$, in the $\lambda \to \infty$ limit. In \cite{han2020pol}, the authors study the same problem on $E_6$ and report three rotationally equivalent classes of solutions - $Para$, $Meta$ and $Ortho$ distinguished by the locations of the splay vertices. This nomenclature originates from organic chemistry to indicate the position of non-hydrogen substituents on a hydrocarbon ring (benzene derivative), and one can identify these substituents with splay or bend vertices. Following the same recipe as in Figure~\ref{fig:dipole}, we compute the limiting profiles for $\binom{6}{2} = 15$ different choices of $\gamma_b$ in \eqref{P_infty}, for $\gamma^* = \pi/6+n\pi$. Discarding rotation and reflection symmetries, we recover at least six different classes of limiting profiles (as opposed to three on $E_6$), plotted in Fig. \ref{fig:dipole_hexagon}. 
On $E^C_6$, there are two classes of $Para$ solutions -  $Para1$ and $Para2$ with two diagonally opposite bend vertices. $Para2$ is distinguished in the sense that the diagonal connecting the bend vertices is orthogonal to the direction $\nvec^* = (\cos \gamma^*, \sin \gamma^*)$. There are two classes of $Meta$ states: $Meta1$ and $Meta2$, for which the two bend vertices are separated by one vertex, and two further classes: $Ortho1$ and $Ortho2$ with two ``adjacent" bend vertices connected by an edge.
The unique $\mathbf{P}$-solution for small enough $\lambda$, is $Para2$ with two diagonally opposite bend vertices on the diagonal orthogonal to the fixed director at infinity; see Fig. \ref{fig:dipole_hexagon} (b). $Para2$ is stable for all $\lambda$. 

\begin{figure*}
\centering
    \begin{subfigure}{0.5\columnwidth}
        \centering
        \includegraphics[width=\textwidth]{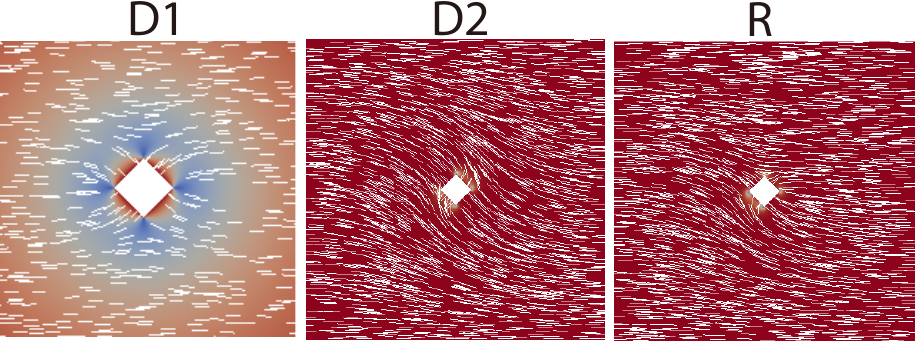}
    \end{subfigure}
    \begin{subfigure}{0.4\columnwidth}
        \centering
        \includegraphics[width=\textwidth]{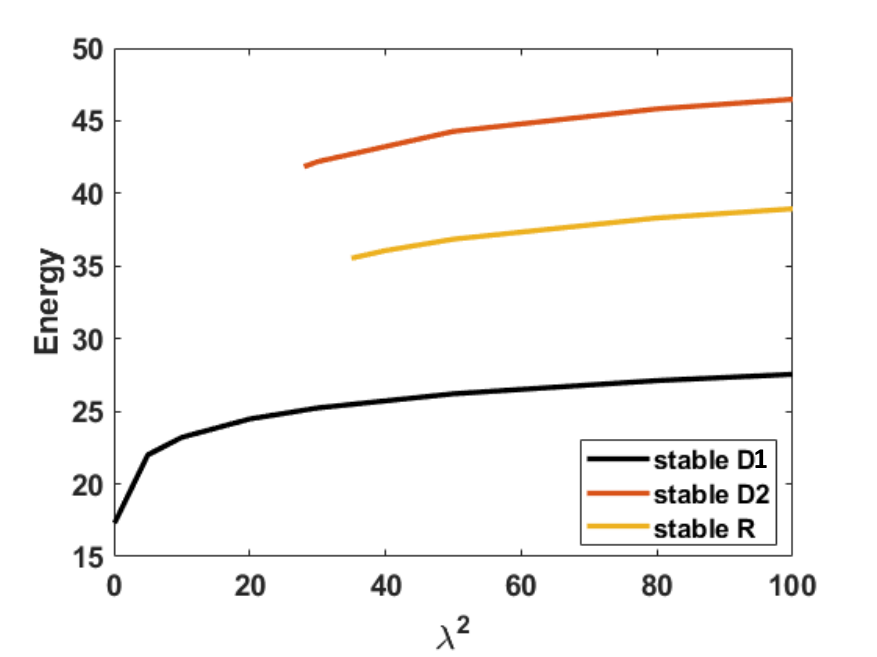}
    \end{subfigure}
        \caption{Left : plots for $D1$ with $\lambda^2 = 5$, $D2$ and $R$ with $\lambda^2 = 100$, the  solutions of the Euler--Lagrange equations in \eqref{Euler_Lagrange} on $[-10,10]^2\backslash E_4$ with $\gamma^* = n\pi$ (the configuration of $D$ is zoomed in). Right: the plot of the rLdG energy in \eqref{p_energy} for states on the left versus $\lambda^2$.
        }
        \label{fig:bifurcation_0}
\end{figure*}

The limiting profiles in Sections~\ref{sec:lambda0} and \ref{sec:infinity} give us excellent initial conditions for finding solution branches of the Euler--Lagrange equations \eqref{Euler_Lagrange}. 
We numerically compute the solutions of \eqref{zero_euler} on domain $[-10,10]^2\backslash E_K$, which is used to approximate the unbounded domain, $E_K^C$, with the Dirichlet boundary condition $\mathbf{P} = \mathbf{P}^*$ in \eqref{constraint_infty} on $x$ or $y= \pm 10$. In  Figures \ref{fig:dipole} and \ref{fig:dipole_hexagon}, we numerically compute the limiting profiles of $\gamma$ for \eqref{infty_euler} on $[-10,10]^2\backslash E_K$. These limiting profiles of $\gamma$ define the limiting profiles, $\mathbf{P}$ via the relation $(P_{11},P_{12}) = \frac{B}{2C}(\cos(2\gamma),\sin(2\gamma))$.
In Figures~\ref{fig:dipole_hexagon}, \ref{fig:bifurcation_0}, \ref{fig:bifurcation}, we use these limiting profiles as initial conditions to compute solution branches for different values of $\lambda^2$; we do not track bifurcations and may have missed solution branches but simply focus on energy comparisons between the numerically computed solution branches. 
We numerically compute the solutions of the Euler--Lagrange equations \eqref{Euler_Lagrange}, subject to the Dirichlet boundary conditions \eqref{Pb}, which are necessarily critical points of \eqref{p_energy}.  We take the solution of \eqref{zero_euler} (\eqref{infty_euler}) on $[-10,10]^2\backslash E_K$ as the initial conditions for finite but small (large) $\lambda^2$. We perform an increasing $\lambda^2$ sweep for the unique branch (in the $\lambda \to 0$ limit) and decreasing $\lambda^2$ sweep for the distinct diagonal- and rotated-branches. For small $\lambda$, a small step-size is chosen and a larger step-size is chosen for large $\lambda$, and this is somewhat adhoc. Except for the $St$ branch in Fig. \ref{fig:bifurcation}, we stop tracking the branch if the smallest eigenvalue approaches zero, i.e. near a bifurcation point.  
For all the numerical simulations in this manuscript including the calculation of the limiting profile on a unit disc in Figures \ref{fig:triangle}, \ref{fig:square}, and \ref{fig:general}, we use the
popular open-source computing software FEniCS for solving partial differential equations, using the finite element method \cite{olgg2012fenics}, which allows us to solve the weak form of the Euler--Lagrange equations or the Laplace equations, in the first order Lagrange element function space by using Newton's method, with tolerance $10^{-13}$ and mesh size less than $1/256$. The convergence may be highly sensitive to the choice of initial condition.
The energy of a solution is calculated according to the integral in \eqref{p_energy}. We check the stability of a solution by checking the sign of the smallest eigenvalue of the Hessian. All computed solution branches are stable in the sense that the corresponding second variation of (\ref{p_energy}) 
\begin{equation}
    \partial^2F_{\lambda}[\eta] = \int_{[-10,10]^2\backslash E_4}|\nabla \eta|^2 + \frac{\lambda^2}{4}\left(|\mathbf{P}|^2-\frac{B^2}{2C^2}\right)|\eta|^2+\frac{\lambda^2}{2}\left(\mathbf{P}\cdot\eta\right)^2 dx
   \label{eq:second}
\end{equation}
is strictly positive according to our numerical computations.
The smallest real eigenvalue of the Hessian is computed by the SLEPc eigenvalue solver with the algorithm krylov-schur, and the tolerance is $10^{-15}$. [https://fenicsproject.org/olddocs/dolfin/1.3.0/python/programmers-reference/cpp/la/SLEPcEigenSolver.html].

In Figure~\ref{fig:bifurcation_0}, we study stable solution branches of the Euler-Lagrange equations \eqref{Euler_Lagrange} on $E^C_4$, with $\gamma^* = n\pi$ and the prescribed boundary conditions on $\partial E_4$, for different values of $\lambda^2$. 
For small $\lambda^2$, there is a unique solution of the system \eqref{Euler_Lagrange} on the domain $[-10,10]^2\backslash E_4$, subject to the boundary conditions \eqref{Pb} and constraint $\gamma^* = n\pi$, labelled as the $D1$ state, with two bend vertices along the diagonal on $x = 0$. This solution branch is computed using the limiting profile in Section~\ref{sec:lambda0}, i.e. the solution of (\ref{zero_euler}) on $E^C_4$ and then using continuation methods for large $\lambda$.
This solution branch remains stable for all $\lambda^2>0$. When $\lambda^2$ is large enough, we observe the stable $D2$ with two bend vertices along the diagonal on $y=0$, and the stable $R$ solution. The $D2$ represents two energetically degenerate diagonal states which are related by reflection about $x = 0$ or $y = 0$ (see the first line of Fig. \ref{fig:dipole}(a)) and the $R$ represent four rotated states (see the second line of Fig. \ref{fig:dipole}(a)). The energy of $D1$ solution is always lower than the $D2$ solution and the rotated solutions have higher energies than the diagonal solutions. We do not investigate the connections between solution branches in this manuscript further.

\begin{figure*}
\centering
    \begin{subfigure}{0.4\columnwidth}
        \centering
        \includegraphics[width=\textwidth]{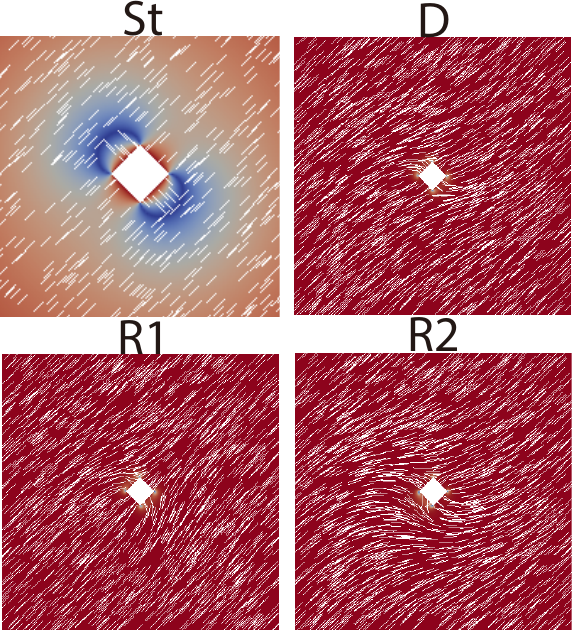}
    \end{subfigure}
    \begin{subfigure}{0.5\columnwidth}
        \centering
        \includegraphics[width=\textwidth]{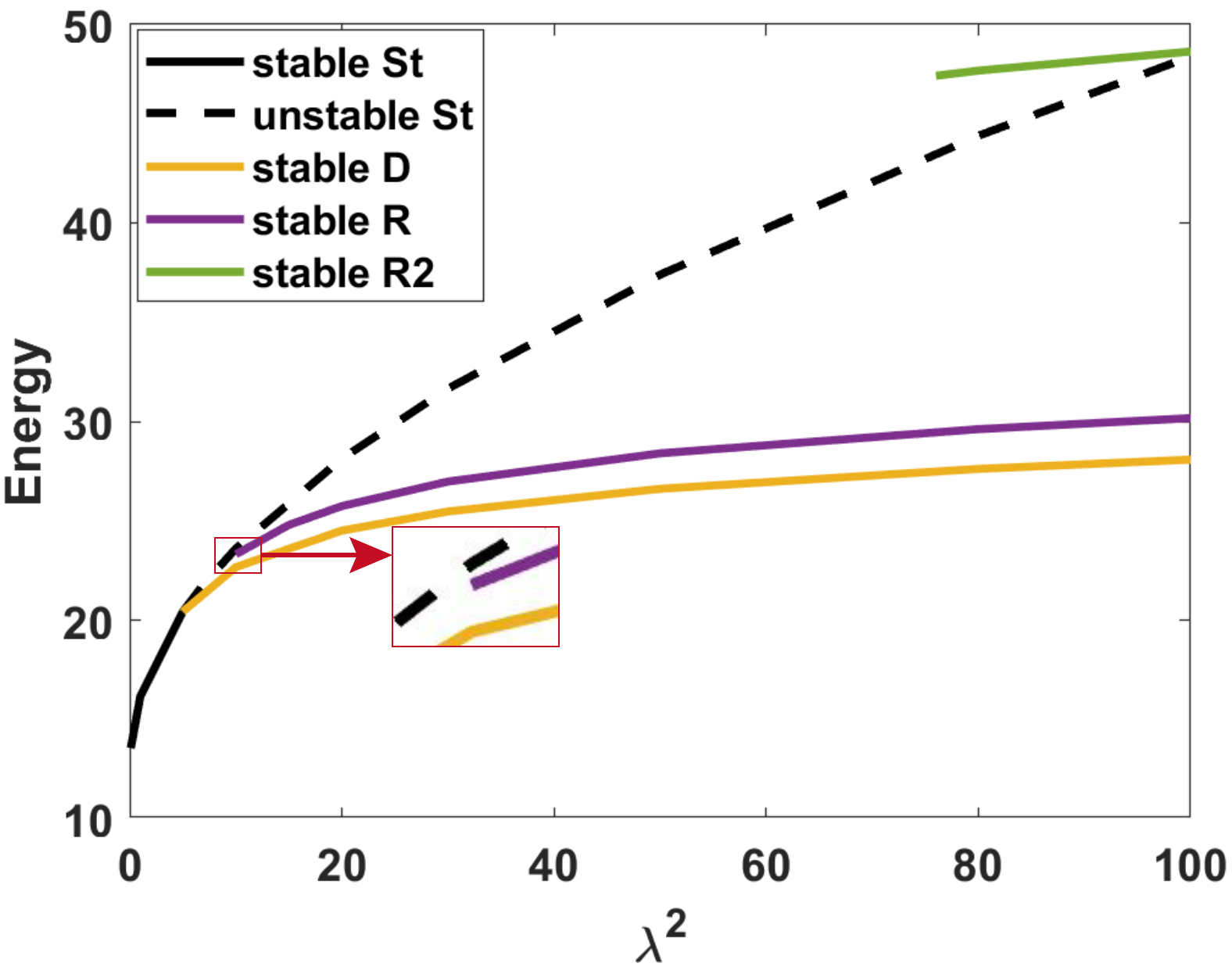}
    \end{subfigure}
        \caption{
         Left : plots for 
         $St$ with $\lambda^2 = 5$, $D$, $R1$ and $R2$ with $\lambda^2 = 100$, the solutions of Euler--Lagrange equations in \eqref{Euler_Lagrange} on $[-10,10]^2\backslash E_4$ with $\gamma^* = \pi/4+n\pi$ (the configuration of $St$ is zoomed in). Right: the plot of energy in \eqref{p_energy} for the states on the left  versus $\lambda^2$, where the solid lines represent stable solution branches and dashed lines correspond to unstable branches. We plot a zoom near the end of $R$, and the branch $R$ doesn't connect with unstable $St$ branch.}
        \label{fig:bifurcation}
\end{figure*}
The case of $\gamma^* = \pi/4+n\pi$ is different (see Fig. \ref{fig:bifurcation}).
The unique $\mathbf{P}$-solution for small $\lambda$, is the $St$-branch with two line defects near opposite square edges.
This $St$-branch exists for all $\lambda^2>0$. As $\lambda^2$ increases, $St$ loses stability and bifurcates into the stable $D$, diagonal solution with two diagonally opposite bend vertices. When $\lambda^2>8.2$, the stable $R1$ solution branch exists. The $R1$ solution has two bend vertices connected by the edge of square, and this edge is perpendicular to $(\cos\gamma^*, \sin\gamma^*)$. There is a stable $R2$ branch, where the $R2$ has two bend vertices along the square edge, which is parallel to $(\cos\gamma^*, \sin\gamma^*)$. This solution branch only exists for $\lambda^2$ large enough. The energy of the diagonal solution is always lower than $R1$, which in turn is energetically preferable to the $R2$ solution branches. 
The families of solutions, labelled by $D$, $R1$ and $R2$, have two members each. The two states in $D$ or $R1$ or $R2$ class are related by reflection about $y = -x$ (see Fig. \ref{fig:dipole}(b)).
The diagram with $\gamma^* = \pi/4+n\pi$ is similar to the diagram on $E_4$ reported in \cite{han2020pol}, with the $St$-branch on $E_4^C$ being analogous to the $WORS$ on $E_4$, and the $D$-branch on $E_4^C$ being analogous to the $D$ on $E_4$, the exception being that the $R1$ and $R2$ solutions are energetically degenerate on $E_4$ but not on $E_4^C$ due to the symmetry-breaking condition at infinity. Similar energy computations are performed for the $Para$, $Meta$, $Ortho$ solution branches in Figure~\ref{fig:dipole_hexagon} on $E_6^C$, for illustrative purposes.

\section{Conclusion}
\label{sec:conclusion}
This paper focuses on illustrative studies of nematic equilibria outside 2D regular polygonal holes, $E_K$ with $K$ edges, focusing on the effects of the shape and size of the polygon in a reduced LdG framework. The equilibria are modelled in terms of a reduced order parameter, $\Pvec$ with two degrees of freedom, that are minimisers of a rLdG free energy (\ref{p_energy}), subject to homeotropic boundary conditions on the hole boundary and a fixed far-field condition, captured by $\gamma^*$. We focus on two distinguished limits: the $\lambda \to 0$ limit relevant for small holes, and the $\lambda \to \infty$ limit relevant for large holes. We study the limiting problem for $\lambda =0$ comprehensively, using methods from complex analysis and Ginzburg-Landau theory. The limiting profile is unique for $\lambda=0$, has an isolated interior point defect for $K=3$, two interior point defects for $K>4$ and for $K=4$, the limiting profile has either line defects or no interior defects depending on $\gamma^*$. In the $\lambda \to \infty$ limit, we have multiple stable equilibria i.e. at least $\binom{K}{2}$ stable equilibria on $E_K^C$, where $K$ is the number of edges of the polygon hole. This is certainly an underestimate, based on two assumptions with regards to minimal topology. Compared to the interior problem studied in \cite{han2020pol}, where the authors study minimisers of the rLdG energy, \eqref{p_energy}, on $E_K$ with tangent boundary conditions, we find that there are more possibilities for defect sets and multistability for the exterior problem. For example, the far-field condition, $\gamma^*$ tunes the location of defects for limiting profiles and for $K=4$, $\gamma^*$ dictates the dimensionality of the defect set for limiting profiles. Further, $\gamma^*$ strongly enhances multistability for the exterior problem compared to the interior counterpart in \cite{han2020pol}. For example, in the $\lambda \to \infty$ limit, there are only two competing equilibria on $E_4$ but there can be up to four different
classes of competing equilibria on $E_4^C$, depending on the choice of $\gamma^*$. In this respect, our work provides good information about how $K$ and $\gamma^*$ can tune the existence, locations and dimensionality of defects on $E_K^C$, and how this can be used to tune multistability.
Our work does not uncover solution landscapes for this exterior problem and does not give information about the connectivity of solution branches i.e. there is no systematic bifurcation analysis. This is much needed and could be tackled using high-index optimisation shrinking dimer methods, as in \cite{hannonlinearity2020}. Such a study would be crucial for understanding pathways between distinct equilibria, the energy barriers and understanding the dynamics of such exterior problems on energy landscapes. In fact, a complete study of solution landscapes can give strong insight into the differences between the interior and exterior problem. Other natural generalisations include addition of the effects of elastic anisotropy (see \cite{hanSIMA2021}), other types of boundary conditions on $\partial E_K$ and studies of arrays of polygonal holes and creation of linked defect structures. Our work provides a foundation for further detailed, comprehensive studies and our overarching goal is to propose universal theoretical frameworks for solution landscapes of confined partially ordered systems, with multiple order parameters.

\textbf{Acknowledgments}
AM and YH thank Ingo Dierking and Adam Draude for suggesting this problem to them, complemented by their experimental work. AM gratefully acknowledges support from the University of Strathclyde New Professors Fund and a University of Strathclyde Global Engagement Grant. AM is also supported by a Leverhulme International Academic Fellowship IAF-2019-009, a Daiwa Foundation Small Grant and a Royal Society Newton Advanced Fellowship. Part of this work was facilitated by a London Mathematical Research Reboot grant awarded to AM, in 2021. YH acknowledges support from a Royal Society Newton International Fellowship. 
\section {Appendix}
\begin{proposition}\label{gamma_infinity_restriction}
We can restrict $\gamma^*\in[0,\frac{\pi}{K}]$, since there are rotation relation between\\ $(p_{11}, p_{12})\vert_{(re^{i\theta+2\pi ki/K},\gamma^*-\frac{2\pi k}{K})}$, $k = 1,\cdots,K$, and $(p_{11}, p_{12})\vert_{(re^{i\theta},\gamma^*)}$, and reflection relation between $(p_{11}, p_{12})\vert_{(re^{i\theta},\gamma^*)}$ and $(p_{11}, p_{12})\vert_{(re^{-i\theta},-\gamma^*)}$.
\end{proposition}
\begin{proof} From the relationship between $\overline{\alpha}_k$ and $\overline{\alpha}_{n+k}$ in \eqref{alpha_rotation}, and the solution $(p_{11},p_{12})$ in \eqref{zero_solution_p11} and \eqref{zero_solution_p12}, one can check
\begin{align}\label{p11_rotation}
&p_{11}\left(re^{i\theta+2\pi n i/K},\gamma^*-\frac{2\pi n}{K}\right)\\
& =  \frac{1}{2\pi}\sum_{k = 1}^K\int_{2\pi(k-1)/K}^{2\pi k/K}\overline{\alpha}_k\frac{1-r^2}{1-2r\cos(\phi-\theta-\frac{2\pi n}{K})+r^2}d\phi + \lim_{\epsilon\to 0}\frac{B}{2C}\frac{\cos(2\gamma^* - \frac{4\pi n}{K})ln r}{ln \epsilon}\nonumber\\
& =  \frac{1}{2\pi}\sum_{k = 1}^{K}\int_{2\pi(k-1-n)/K}^{2\pi (k-n)/K}\overline{\alpha}_k\frac{1-r^2}{1-2r\cos(\phi-\theta)+r^2}d\phi + \lim_{\epsilon\to 0}\frac{B}{2C}\frac{\cos(2\gamma^* - \frac{4\pi n}{K})ln r}{ln \epsilon}\nonumber\\
& = \frac{1}{2\pi}\sum_{k = 1-n}^{K-n}\int_{2\pi(k-1)/K}^{2\pi k/K}\overline{\alpha}_{k+n}\frac{1-r^2}{1-2r\cos(\phi-\theta)+r^2}d\phi + \lim_{\epsilon\to 0}\frac{B}{2C} \frac{\cos(2\gamma^* - \frac{4\pi n}{K})ln r}{ln \epsilon}\nonumber\\
& = \frac{1}{2\pi}\sum_{k = 1}^{K}\int_{2\pi(k-1)/K}^{2\pi k/K}\overline{\alpha}_{k+n}\frac{1-r^2}{1-2r\cos(\phi-\theta)+r^2}d\phi + \lim_{\epsilon\to 0}\frac{B}{2C} \frac{\cos(2\gamma^* - \frac{4\pi n}{K})ln r}{ln \epsilon}\nonumber\\
&=\cos\left(\frac{4\pi n}{K}\right)p_{11}(re^{i\theta},\gamma^*) + \sin\left(\frac{4\pi n}{K}\right)p_{12}(re^{i\theta},\gamma^*).
\end{align}
In the first instance, we change the range of integration of $\phi$ using a change of variable, and then use (\ref{alpha_rotation}) to write $\overline{\alpha}_{k+n}$ in terms of $\overline{\alpha}_k$ and $\overline{\beta}_k$ etc.
Using analogous arguments, one obtains
\begin{equation}\label{p12_rotation}
p_{12}\left(re^{i\theta+2\pi ni/K},\gamma^*-\frac{2\pi n}{K}\right) = \cos\left(\frac{4\pi n}{K}\right)p_{12}(re^{i\theta},\gamma^*) - \sin\left(\frac{4\pi n}{K}\right)p_{11}(re^{i\theta},\gamma^*).
\end{equation}

Using the relation between $\overline{\alpha}_k$ and $\overline{\alpha}_{K-k+1}$ in \eqref{alpha_reflection}, we have
\begin{align}\label{p11_reflection}
&p_{11}(re^{-i\theta},-\gamma^*)\\ 
& = \frac{1}{2\pi}\sum_{k = 1}^K\int_{2\pi(k-1)/K}^{2\pi k/K}\overline{\alpha}_k\frac{1-r^2}{1-2r\cos(\phi+\theta)+r^2}d\phi + \lim_{\epsilon\to 0}\frac{B}{2C}\frac{\cos(-2\gamma^*)ln r}{ln \epsilon}\nonumber\\
& = \frac{1}{2\pi}\sum_{k = 1}^K\int_{2\pi(k-1)/K}^{2\pi k/K}\overline{\alpha}_k\frac{1-r^2}{1-2r\cos(-\phi-\theta)+r^2}d\phi + \lim_{\epsilon\to 0}\frac{B}{2C}\frac{\cos(2\gamma^*)ln r}{ln \epsilon}\nonumber\\
& = -\frac{1}{2\pi}\sum_{k = 1}^K\int_{-2\pi(k-1)/K}^{-2\pi k/K}\overline{\alpha}_k\frac{1-r^2}{1-2r\cos(\phi-\theta)+r^2}d\phi + \lim_{\epsilon\to 0}\frac{B}{2C}\frac{\cos(2\gamma^*)ln r}{ln \epsilon}\nonumber\\
& = -\frac{1}{2\pi}\sum_{k = 1}^K\int_{2\pi (K-k+1)/K}^{2\pi (K-k)/K}\overline{\alpha}_k\frac{1-r^2}{1-2r\cos(\phi-\theta)+r^2}d\phi + \lim_{\epsilon\to 0}\frac{B}{2C}\frac{\cos(2\gamma^*)ln r}{ln \epsilon}\nonumber\\
& = \frac{1}{2\pi}\sum_{k = 1}^K\int_{2\pi (K-k)/K}^{2\pi (K-k+1)/K}\overline{\alpha}_k\frac{1-r^2}{1-2r\cos(\phi-\theta)+r^2}d\phi + \lim_{\epsilon\to 0}\frac{B}{2C}\frac{\cos(2\gamma^*)ln r}{ln \epsilon}\nonumber\\
& = p_{11}(re^{i\theta},\gamma^*)
\end{align}
and using analogous arguments, we obtain
\begin{equation}\label{p12_reflection}
p_{12}(re^{-i\theta},-\gamma^*) = -p_{12}(re^{i\theta},\gamma^*).
\end{equation}
\end{proof}


\begin{corollary}\label{piK_symmetry}
With $\gamma^* = \pi/K$, $(P_{11},P_{12})$ has reflection symmetry about $\psi = \pi/K$, i.e.,
\begin{align}
P_{11}(\rho e^{\pi/K i-\psi i},\pi/K) &= P_{11}(\rho e^{\pi/K i+ \psi i},\pi/K)\cos(4\pi/K) + P_{12}(\rho e^{\pi/K i+ \psi i},\pi/K)\sin(4\pi/K),\\
P_{12}(\rho e^{\pi/K i-\psi i},\pi/K) &= -P_{12}(\rho e^{\pi/K i+ \psi i},\pi/K)\cos(4\pi/K) + P_{11}(\rho e^{\pi/K i + \psi i},\pi/K)\sin(4\pi/K).
\end{align}
\end{corollary}
\begin{proof}
Use \eqref{p11_rotation} and \eqref{p12_rotation} with $\gamma^* = -\pi/K$, $n = -1$, and substitute $\theta+\pi/K$ for $\theta$ to get
\begin{align}
&p_{11}\left(r e^{-\pi/K i +\theta i},\pi/K\right) = \cos\left(\frac{-4\pi}{K}\right)p_{11}(re^{i\pi/K + i\theta},-\pi/K) + \sin\left(\frac{-4\pi}{K}\right)p_{12}(re^{i\pi/K + i\theta},-\pi/K),\label{11}\\
&p_{12}\left(r e^{-\pi/K i +\theta i},\pi/K\right) = \cos\left(\frac{-4\pi}{K}\right)p_{12}(re^{i\pi/K + i\theta},-\pi/K) - \sin\left(\frac{-4\pi}{K}\right)p_{11}(re^{i\pi/K + i\theta},-\pi/K).\label{12}
\end{align}
Combining \eqref{p11_reflection} with \eqref{p12_reflection}, it follows that
\begin{gather}\label{2}
p_{11}\left(r e^{\pi/K i +\theta i},-\pi/K\right) = p_{11}\left(r e^{-\pi/K i-\theta i},\pi/K\right),\qquad p_{12}\left(r e^{\pi/K i+\theta i},-\pi/K\right) = -p_{12}\left(r e^{-\pi/K i-\theta i},\pi/K\right).
\end{gather}

Substituting \eqref{2} into \eqref{11} and \eqref{12}, we obtain
\begin{align}
&p_{11}\left(r e^{-i\pi/K +i\theta},\pi/K\right) = \cos\left(\frac{4\pi}{K}\right)p_{11}(re^{-i\pi/K -i\theta},\pi/K) + \sin\left(\frac{4\pi}{K}\right)p_{12}(re^{-i\pi/K - i\theta},\pi/K),\\
&p_{12}\left(r e^{-i\pi/K +i\theta},\pi/K\right) = -\cos\left(\frac{4\pi}{K}\right)p_{12}(re^{-i\pi/K - i\theta},\pi/K) + \sin\left(\frac{4\pi}{K}\right)p_{11}(re^{-i\pi/K - i\theta},\pi/K).
\end{align}
Using \eqref{f_reflection} and \eqref{f_rotation}, the SC mapping in \eqref{eq:SC} preserves reflection symmetry, $f(re^{-i\pi /K-i\theta }) = \overline{f(re^{i \pi /K + i\theta})} = e^{i2\pi /K}\overline{f(re^{-i\pi/K + i\theta})}$.
We have for any $\rho e^{i\pi/K + \psi i}\in E_K^C$ satisfying $f(re^{-i\pi/K - i\theta}) = \rho e^{i\pi/K + i\psi }$ and $f(re^{-i\pi/K + i\theta}) = \rho e^{i\pi/K - i\psi }$, $(P_{11},P_{12})(\rho e^{i\pi/K + i\psi},\pi/K) = (p_{11},p_{12})(re^{-i\pi/K - i\theta},\pi/K)$ and $(P_{11},P_{12})(\rho e^{i\pi/K - i\psi },\pi/K) = (p_{11},p_{12})(re^{-i\pi/K + i\theta},\pi/K)$. Hence $(P_{11},P_{12})$ has reflection symmetry about $\psi = \pi/K$.
\end{proof}

\newpage
\bibliographystyle{unsrt}
\bibliography{exterior_arxiv.bib}
\end{document}